\documentclass[aps,prx,showkeys,twocolumn]{revtex4-2}

\usepackage{amsmath,amsfonts,amssymb,amsthm,stmaryrd,amsbsy,bm}
\usepackage{mathrsfs,mathtools,empheq}
\usepackage{scalerel}
\usepackage{tcolorbox}
\usepackage{graphicx}
\usepackage{hyperref}
\usepackage{enumitem}
\usepackage{accents}
\usepackage{natbib}
\usepackage{tikz-cd}
\usepackage{multirow}
\makeatletter
\def\squiggly{\bgroup \markoverwith{\textcolor{red}{\lower3.5\p@\hbox{\sixly \char58}}}\ULon}
\makeatother

\numberwithin{equation}{section}

\theoremstyle{plain}
 \newtheorem{thm}{Theorem}[section]
 \newtheorem{lem}[thm]{Lemma}
 
\newtheorem{prop}[thm]{Proposition}

 \newtheorem*{lem*}{Lemma}
 \newtheorem*{cor*}{Corollary}
 \newtheorem*{thm*}{Theorem}
 \theoremstyle{definition}
 \newtheorem{defn}{Definition}[section]
 \theoremstyle{remark}

\newcommand{\Reals}{{\mathbb R}}    % Reals
\newcommand{\Cmplxs}{{\mathbb C}}   % Complexes
\newcommand{\evec}[1]{{\bm e}_{#1}}  % unit vector
\newcommand{\inpr}[2]{\left\langle {#1} \,\middle|\, {#2} \right\rangle}    % inner product
\newcommand{\ilinpr}[2]{\langle {#1} \,\mid\, {#2} \rangle}    % inner product
\newcommand{\defeq}{ {\kern 0.2em}:{\kern -0.5em}={\kern 0.2 em} }  % equal by def
\newcommand{\eqdef}{ {\kern 0.2em}={\kern -0.5em}:{\kern 0.2 em} }  % equal by def (backward)
\newcommand{\setof}[2]{\left\{ #1 \;\middle|\; #2 \right\}}  % set-builder notation
\newcommand{\Arr}[3]{{#2}\,\colon\, {#1} \rightarrow {#3}} % morphisms
\DeclareMathOperator{\ran}{ran}  % ran
\newcommand{\Id}{\mathrm{id}}  % identity operations
\newcommand{\dom}{\mathrm{dom}\,}

\newcommand{\Prob}{P}
\newcommand{\Expct}{\mathsf E}

\DeclareMathOperator{\Var}{\mathrm{var}}

\newcommand{\restrict}[2]{{#1}{\kern -0.3em}\upharpoonright{\kern -0.3em}{#2}}  % 
\DeclareMathOperator{\supp}{supp}
\DeclareMathOperator{\Span}{span}

\newcommand{\Site}{\mathsf{Site}} % index set   !!
\newcommand{\Sps}[1]{\mathsf{Sps}(#1)} % species   !!
\newcommand{\Cfg}{\mathsf{Cfg}}  % configuration space   !!
\newcommand{\Fns}{{\mathcal F}}

\DeclareMathOperator{\Av}{{\mathsf{Av}}}
\newcommand{\EOp}[1]{{\mathsf E}_{#1}}
\newcommand{\COp}[1]{{\mathsf C}_{#1}}

\newcommand{\Clust}{{\mathcal C}}

\DeclareMathOperator{\Mop}{\mathscr{M}}

\newcommand{\IS}{{\bf IS}}
\newcommand{\Gp}{G} %{{\mathscr G}}

\newcommand{\Smpl}{{\ensuremath{Smpl}}}
\newcommand{\Small}{{\ensuremath{Smll}}}

\newcommand{\x}{x}

\newcommand{\wh}[1]{\widehat{#1}}
\newcommand{\isit}[1]{\delta_{#1}}
\newcommand{\Up}[2]{\Lambda_{#1}^{#2}}
\newcommand{\ssp}[1]{{\mathcal #1}}

\newcommand{\prS}{{\mathsf S}} 
\newcommand{\prV}{{\mathsf V}} 
\newcommand{\prM}{{\mathsf M}} 
\DeclareMathOperator{\fold}{\mathsf{fold}}

\newcommand{\hide}[1]{}
\newcommand{\Ax}[1]{{\textbf{#1}}}
\renewcommand{\implies}{\ensuremath{\;\Rightarrow\;}}

\renewcommand{\evec}{\vec{e}}
\newcommand{\vals}{{\mathcal V}}
\newcommand{\ClSh}{\mathsf{ClSh}}

\DeclareFontFamily{U}{matha}{\hyphenchar\font45}
\DeclareFontShape{U}{matha}{m}{n}{
      <5> <6> <7> <8> <9> <10> gen * matha
      <10.95> matha10 <12> <14.4> <17.28> <20.74> <24.88> matha12
      }{}
\DeclareSymbolFont{matha}{U}{matha}{m}{n}
\DeclareFontFamily{U}{mathx}{\hyphenchar\font45}
\DeclareFontShape{U}{mathx}{m}{n}{
      <5> <6> <7> <8> <9> <10>
      <10.95> <12> <14.4> <17.28> <20.74> <24.88>
      mathx10
      }{}
\DeclareSymbolFont{mathx}{U}{mathx}{m}{n}

\DeclareMathSymbol{\obot}         {2}{matha}{"6B}
\DeclareMathSymbol{\bigobot}       {1}{mathx}{"CB}
\begin{document}
%\pagecolor{brown!22}
%%%%%%%%%%%%%%%%%%%%%%%%%%%%%%%%%%%%%%
\title{Cluster expansion methods from physical concepts}
\author{Paul E. Lammert}
\email{lammert@psu.edu}
\affiliation{Department of Physics,
Pennsylvania State University, University Park, PA 16802-6300}
\author{Vincent H. Crespi}
\affiliation{Department of Physics,
Pennsylvania State University, University Park, PA 16802-6300}
\affiliation{Department of Chemistry, 
Pennsylvania State University, University Park, PA 16802-6300}
\affiliation{Department of Materials Science and Engineering 
Pennsylvania State University, University Park, PA 16802-6300}

%%%%%%%%%%%%%%%%%%%%%%%%%%%%%%%%%%%%%%
\begin{abstract}
  The cluster expansion formalism used in materials science is reconstructed on
  an axiomatic basis with the aims of clarifying underlying concepts
  and improving computational procedures, and without using conventional
  cluster functions.
Instead, cluster components of configuration functions are defined in an
intrinsic manner, which can be viewed as M\"{o}bius inversion of conditional
expectation.
The associated method for fitting a model to a configurational sample
is grounded entirely in Hilbert space geometry.
By constructing models directly from the given data, we avoid an
underdetermination problem to which the conventional approach is subject.
Tensor observables are treated on an equal footing with scalar observables.
\end{abstract}
%%%%%%%%%%%%%%%% 
\date{Aug. 30, 2022}
\maketitle
\tableofcontents

\section{Introduction}

Discovery and elucidation of structure-property relationships is a core
concern of materials science\cite{van-de-Walle-08}.
For dealing with configurational dependence,
cluster expansion methods\cite{Sanchez+84,de-Fontaine-94,Zunger-94,Sanchez-10,Sanchez-17b,Wolverton+de-Fontaine-94,van-de-Walle+02,Lerch+09,Angqvist+19,Chang+19} have become a widespread tool, as attested by recent
reviews\cite{Wu+16,Van-der-Ven+18,Kadkhodaei+Munoz-21}.
There are now many software packages and
libraries\cite{ICET,ATAT,CASM,CELL,CLEASE}
allowing one to incorporate cluster expansion into multiscale modeling.
We suggest that, despite this level of development, the core formalism is not
understood as well as it should be.
A symptom of this is rigid adherence to specific, pre-defined,
orthonormal bases of so-called {\it cluster functions} $\Phi_I^\alpha$, so that
one cannot describe what is going on in a basis-free manner.
This is like distinguishing {\it p}\/-orbitals
not by using the concept of angular momentum, but through a fixed $p_x$, $p_y$
and $p_z$ basis. From a practical standpoint, the problem with such a situation is
that, internally, one may fail to see opportunities to do computations
more efficiently or robustly,
while externally, distinguishing cluster expansion from other potential ways
of solving the modeling problem (see paragraph after next) or even conceiving possible
contours of such a thing are more difficult. 
For example, imagine some seemingly new of machine learning algorithm.
How could we tell whether it was really a sort of cluster expansion in disguise?

The project of this paper is first to unearth and clearly display the logic of
cluster expansion methods --- what do they mean? Then we show the practical
utility of our answer.
In addition to its intrinsic interest, this supports the claim that the
proposed view of matters is a specially good one.
The fundamental issue is to understand whether there is some
well-defined and nonarbitrary principle whereby a configuration function
is broken down into elementary building blocks like $n$-body interactions.
We show that there is such a principle (See Def.~\ref{def:cluster-decomposition}).
However, the decomposition is not absolute, but relative to a probability distribution on
configurations.
A Hilbert space structure (i.e., inner product), an unexplained basic ingredient of the
conventional formalism, appear as a natural consequence of this conceptual analysis.
Section~\ref{sec:independence} is then able to explain what makes the
conventional cluster functions $\Phi_I^\alpha$ special --- they provide bases
of structurally defined cluster subspaces when sites are probabilistically independent.

Return to the {\it modeling problem} mentioned earlier.
Given total energy on a sampling of configurations, how can we find
an approximate decomposition into zero-body, one-body, two-body, \dots\ interactions?
This problem (not only for energy) makes up the great majority of applications
of cluster expansion methods.
Our position is that such modeling should explicitly reflect and incorporate the
Hilbert space structure of the space of configuration functions.
The conventional approach does not do that, but rather
fits coefficients $J_I^\alpha$ of a model $\sum J_I^\alpha \Phi_I^\alpha$
as though the $\Phi_I^\alpha$ were generic expansion functions.
In contrast, we treat the problem from a perspective based on Hilbert space geometry.
A model is expanded on a basis derived directly from the sample data,
thereby avoiding an underdetermination problem to which the conventional approach
is prone.

\subsection*{Outline}

We give now a more detailed guide to the contents.
But first, some general comments.
It may appear at first glance that the paper builds up a huge and unwieldy apparatus.
However, much of it is one-time analysis, with the goal of a clear and {\em flexible}
formalism. We also pay attention to computational complexity.
The mathematics involved is mostly discrete probability and
finite-dimensional vector/Hilbert spaces, theory covered in many
quantum mechanics courses such as \onlinecite{Sakurai-QM,Bohm-QM}.
Sections~\ref{sec:symmetry} and \ref{sec:efficiency}
involve a little general group representation theory\cite{Lax,Dresselhaus}.
Generally, it is safe to skip things enclosed in ``{\it Proof} \ldots $\qedsymbol$\ ''.

Section~\ref{sec:critical-review} is a brief review of the conventional
cluster expansion formalism. It expands on the motivations already given
in this Introduction.
Section~\ref{sec:notation} establishes notation which will be used throughout.
Section~\ref{sec:from-scratch} begins our new development from the ground up,
introducing the ``best approximation'' idea and showing
how elementary building blocks (cluster components) can be extracted.
At least from a logical perspective,
Section~\ref{sec:axioms-and-consequences} is the core of the new formulation.
We give informal arguments for a small set of axioms, the idea being to sharpen
and codify an initially inchoate concept of best approximation in this context
of arbitrary configuration functions.
Unpacked, the axioms say we have a probability distribution on configurations.
This is not surprising, and possibly even reassuring.
Very importantly for what follows, this section sets up the naturally
associated Hilbert space.

The conventional formalism requires probabilistic independence of
the degrees of freedom on distinct sites.
Section~\ref{sec:independence} explores this restriction, developing a
long list of equivalents and consequences.
From that point on, we assume that the independence condition holds.
Section \ref{sec:Indep-defense} explains that without that restriction, cluster
components have properties which are probably counter to intensions.
Site independence is therefore presented as an {\it secondary axiom}.
Section~\ref{sec:Phis} recovers the conventional cluster functions
in asking for orthonormal bases (ONBs) of the cluster subspaces.
At that point, the core formalism is finished.
Section~\ref{sec:intermission} recapitulates and does some stage-setting
for the more applied part which follows.
Section~\ref{sec:symmetry} brings both space-group symmetry and
tensor-valued observables into the formalism; it is natural to
consider these together. Many important material properties are tensorial,
and they are incorporated into our framework with no more difficulty than
than into the conventional one\cite{van-de-Walle-08}.
Section~\ref{sec:efficiency} shows that eschewing cluster functions
is computationally reasonable.
Finally, section~\ref{sec:modeling} tackles the modeling problem
referred to in the second paragraph of this introduction.
Our approach is exclusively through Hilbert space geometry.
We show that it is practical to construct a model directly from
the sampled data, and that this even avoids an underdetermination
problem which affects the conventional formalism.

\section{Critical review of conventional formalism}
\label{sec:critical-review}

This Section gives a brief overview of the conventional cluster expansion formalism.
The purposes are to amplify points made earlier, make the contrast with our
approach clearer, and help those familiar with the conventional methods to avoid
confusion when we do something different. 
In principle, none of this is crucial for what follows.
The new formalism is developed from the ground up in the following sections.

\subsection{Apparatus}\label{sec:conventional-apparatus}

The basic context for cluster expansion methods is a finite collection
$\Site$ of {\it sites}, and a collection $\mathsf{Sps}$ of {\it species} that may
occupy the sites. 
Typically $\mathsf{Sps}$ is something like \{Ni,Zn,Mg,Cu,Co\} and the sites are
sites of a (finite or periodic) crystal structure.
A configuration $x$ is an assignment $i\mapsto x_i$ of species to sites.
The set of all configurations is denoted $\Cfg$.
The subject that cluster expansion talks about is real or complex
configuration functions. These comprise a vector space $\Fns$. 

There are two ubiquitous pieces of machinery in the cluster expansion literature:
encoding by numbers, and the use of so-called cluster functions as a basis for
$\Fns$. The former is really a somewhat superficial aspect, but we discuss it
anyway due to its ubiquity and the influence it has on the cluster functions.

\subsubsection{encoding}\label{sec:sketch-encoding}

Encoding is almost universally presented as a matter of replacing
$\mathsf{Sps}$ with a set of numbers. Typically, with our example $\mathsf{Sps}$, we would have
\hbox{Ni $\mapsto 1$}, \ldots, \hbox{Co $\mapsto 5$}, or 
\hbox{Ni $\mapsto -2$}, \ldots, \hbox{Co $\mapsto 2$}.
Each site $i$ is then equipped with a
{\it site variable}, {\it occupation variable}, or {\it spin variable}
(common terms), such that, for example, $\sigma_i = 1$ means that there is Ni at site $i$.

We believe that, although formally it changes almost nothing,
a better perspective is to consider $\sigma_i$ not as a variable,
but as a {\em fixed} function, an element of $\Fns$, because of the way the
$\sigma_i$'s are used.
$\Fns$ is not only a vector space, it is an algebra, that is, its members
can be multiplied together: $(fg)(x) = f(x)g(x)$.
The algebra generated by the $\sigma_i$'s, that is, the set of polynomials
in them collectively, is necessarily a subset of $\Fns$.
If each $\sigma_i$ is one-to-one, this algebra is all of $\Fns$.
The key point is the following lemma.
The $*$ appearing in it can be thought of as standing for
a generic, uspecified, or abstract site.
A proof is given for completeness, but none of the ideas therein are
used anywhere else.
\begin{lem}\label{lem:encoding}
  Suppose $\sigma_*$ is a one-to-one real (or complex) function on
  a set $\mathsf{Sps}$ with $N$ elements ($|\mathsf{Sps}| = N$).
  Then, the powers $1=\sigma_*^0,\sigma_*,\sigma_*^2,\ldots,\sigma_*^{N-1}$ are
  a linear basis for the space $\Fns(*)$ of functions $\mathsf{Sps}\rightarrow \Reals$.
  For some polynomial $p$ of degree $N-1$, $\sigma_*^N = p(\sigma_*)$.
\end{lem}

Thus, for each site $i$, the space $\Fns(i)$ of functions depending on only
the species at site $i$
is a copy of $\Fns(*)$, and $1,\sigma_i,\ldots,\sigma_i^{|\mathsf{Sps}|-1}$ is
a basis of $\Fns(i)$.
Furthermore, the set of all monomials $\prod_{i\in\Site} \sigma_i^{m_i}$,
with each $m_i$ in $\{0,1,\ldots,|\mathsf{Sps}|-1\}$, is therefore a linear
basis of $\Fns$. 
Since factors $\sigma_i^0$ make no difference to the product, we
can equally describe this basis as consisting of monomials $\prod_{i\in I} \sigma_i^{m_i}$
for some set of sites $I$, and with $1 \le m_i \le |\mathsf{Sps}|-1$.

Any member of $\Fns$ can be expressed as a polynomial in the $\sigma_i$'s,
this representation is highly nonunique. Although this fact has little effect on
most applications, it certainly complicates theoretical considerations, nor
does it seem to us to have practical advantages over dealing directly with
configuration functions.

\begin{proof}[Proof of Lemma~\ref{lem:encoding}]
  Let $\mathsf{Sps} = \{s_1,\ldots,s_N\}$, and $\delta_i$ be the indicator function of
  $s_i$ ($\delta_i(s_j)$ equals 1 if $j=i$, otherwise zero).
  (1) The dimension of $\Fns(*)$ is $N$ because the $\delta_i$'s are a basis.
(2) 
Since $\delta_1$ is a multiple of
$[\sigma_* - \sigma_*(s_2)] \cdots [\sigma_* - \sigma_*(s_N)]$, and similarly
for the other indicators, polynomials in $\sigma_*$ span $\Fns(*)$.
Hence, any linearly independent collections of polynomials can be expanded to
one of cardinality $N$, but no larger.
(3) Now, let $M$ be smallest such that $\sigma_*^M$ is linearly dependent on
lower powers, so $\sigma_*^0,\sigma_*^1,\ldots,\sigma_*^{M-1}$ are
linearly independent and for some polynomial $p$ of degree $M-1$,
$\sigma_*^M = p(\sigma_*)$.
Repeated use of this equation allows allows to express also every higher
power of $\sigma_*$ as a linear combination of $\sigma_*^0,\ldots,\sigma_*^{M-1}$.
Hence these must form a basis, and $M=N$.
\end{proof}

\subsubsection{orthonormal bases of cluster functions}
\label{sec:sketch-cluster-functions}

The most straightforward way to obtain the standard orthonormal bases
(ONBs) of {\it cluster functions}, and one common in the literature,
is to treat $\Fns$ as a tensor product of Hilbert spaces.
Give the vector space $\Fns(i)$ of functions of the configuration variable
at site $i$ an inner product according to
\begin{equation}
  \label{eq:conventional-inner-prod-0}
  \inpr{h}{g}_i
  = \frac{1}{|\mathsf{Sps}|}\sum_{x_i\in\mathsf{Sps}} h(x_i)g(x_i).
\end{equation}
This turns $\Fns(i)$ into a Hilbert space.
Now regarding $\Fns$ as the tensor product $\bigotimes_i \Fns(i)$, it acquires
the inner product (recall, $\Cfg$ is the configuration space)
\begin{equation}
  \label{eq:conventional-inner-prod}
\inpr{h}{g} = \frac{1}{|\Cfg|}\sum_{x\in\Cfg} h(x)g(x).
\end{equation}

Typically, the notion of tensor product is explained in terms of
orthonormal bases (ONBs).
Returning to the linear basis $1,\sigma_*,\ldots,\sigma_*^{|\mathsf{Sps}|-1}$ of 
the generic single-site space $\Fns(*)$, construct an ONB by Gram-Schmidt
orthogonalization.
This gives orthonormal functions $\varphi_k(\sigma_*)$ for
\hbox{$k=0,1,\ldots,|\mathsf{Sps}|-1$},
where $\varphi_k$ is a degree-$k$ polynomial.
All possible products $\prod \varphi_{k(i)}(\sigma_i)$ over all sites
then constitutes an ONB for $\Fns$. Factors of $\varphi_0$
can be ignored, since they equal the constant one, so we can also organize this
ONB as all products
\begin{equation}
  \label{eq:Phi}
\Phi_I^\alpha = \prod_{i\in I} \varphi_{\alpha(i)}(\sigma_i),  
\end{equation}
where $I$ runs over sets of sites and $\alpha$ assigns a {\em nonzero} index
to each site in $I$. We will call the set $I$ the {\it support} of $\Phi_I^\alpha$.
If $I$ is empty, then there are no factors in the product.
The functions $\Phi_I^\alpha$ are often called {\it cluster functions},
though sometimes simply {\it basis functions}. Our notation is close to that
of Ref. \onlinecite{Wolverton+de-Fontaine-94}.
Other common notations are $\Phi_\alpha$, where the support is left implicit
in $\alpha$, or $\Pi_\alpha$, suggesting a product.
We prefer to make the support explicit and leave the range of $\alpha$ implicit.

%%%%%%%%%%%%%%%
\subsection{Shortcomings}\label{sec:criticism}

A Hilbert space structure, that is an inner product,
has been introduced here, along with a supposedly distinguished
kind of orthonormal basis of cluster functions.
However, it is unclear where the inner product comes from
or in precisely what way the $\Phi_I^\alpha$ are distinguished or are related to
the intuitive idea of an $n$-body interaction.
Even if the structure is judged reasonable, a deeper explanation is desirable
for added flexibility and a clearer idea of the possibility of alternatives.

Certainly, the cluster functions have an arbitrariness coming from the choice of the
$\varphi_k$, or of $\sigma_*$.
It might seem plausible that this arbitrariness is limited to $\alpha$, so
that the sum
\begin{equation}
\nonumber
\COp{I} f \defeq \sum_\alpha \inpr{\Phi_I^\alpha}{f} \Phi_I^\alpha,  
\end{equation}
actually depends only on the Hilbert space structure.
If so, it has intrinsic meaning and characterization, but what are they?
We claim that it should be considered the ``$I$-body'' part (as in two-body,
three-body, \dots) of $f$. This characterization will emerge
Section~\ref{sec:axioms-and-consequences}.

There is a generalization of the simple form of the conventional formalism
which has been considered in the literature\cite{Sanchez-93,Sanchez-10,Sanchez-17b}.
This is to replace the inner product (\ref{eq:conventional-inner-prod-0}) by
\begin{equation}
\nonumber 
\inpr{h}{g}_i = \sum_{x_i\in\mathsf{Sps}} p(x_i) h(x_i) g(x_i),
\end{equation}
where the weights $p(x_i)$ add to one, and therefore can be interpreted
as probabilities. This effectively makes the inner product a derived notion
based on the probability distribution $p$.
We will go even farther and not even start with a probability distribution
on configurations.

Any observable, energy ${\mathcal E}$ for instance, has an expansion
over the cluster functions as
\begin{equation}
  \label{eq:conventional-clstr-xpn}
{\mathcal E}(x) = \sum_{I,\alpha} \inpr{\Phi_I^\alpha}{f} \Phi_I^\alpha(x).  
\end{equation}
It is common to use the term \textit{cluster expansion} to refer to the
construction of an approximation
\begin{equation}
\label{eq:approx-clstr-xpn}
{\mathcal E} \approx \sum_{\Small\, I, \alpha} J_I^\alpha \Phi_I^\alpha
\end{equation}
based on the energies of a limited sample of configurations.
In the case of energy, the $J_I^\alpha$'s are called effective
cluster interactions (ECIs).
It is reasonable to call the right-hand side of either
(\ref{eq:conventional-clstr-xpn}) or (\ref{eq:approx-clstr-xpn})
a cluster expansion.
However, the conventional approach to determining the model parameters (ECIs)
takes no particular notice of the fact that the $\Phi_I^\alpha$'s
are orthonormal with respect to a presumably relevant inner product.
The methods are the same as would be deployed for expansion in a generic
set of functions.
Thus, the conventional modeling procedures are not
hardly part of cluster expansion formalism proper. 
This paper takes the perspective that the Hilbert space structure of
$\Fns$ has implications for modeling that should be taken seriously.
Section~\ref{sec:modeling} shows how. Modeling done that way can
be considered a proper part of the cluster expansion formalism.

%%%%%%%%%%%%%%%%
\section{Notation and nomenclature}\label{sec:notation}

\begin{table}[h]
\begin{tabular}{ll}
$i,i',j,j',\ldots$ & site \\
$\Sps{i}$ & set of possible species at site $i$ \\
$\Site$  & set of all sites \\
$I,J,K,\ldots$ & subset of $\Site$ (cluster) \\
$x,x',y,y',\ldots$ & [partial] configuration ($x_i\in\Sps{i}$) \\
$\Cfg$ & set of total configurations \\
$\Cfg(I)$ & set of partial configurations defined on $I$ \\
$f,g,\ldots$ & observable \\
$\Arr{\Cfg}{\Fns}{\Cmplxs}$ & space of observables \\
% $\Fns$ & space $\Cfg \rightarrow \Cmplxs$ of observables \\
$\Arr{\Cfg(I)}{\Fns(I)}{\Cmplxs}$  & functions supported on $I$ \\
%$\Fns(I)$ & functions supported on $I$, $\Cfg(I)\rightarrow \Cmplxs$ \\
$\supp$ & support \\
  $\Arr{\Fns}{\EOp{I}}{\Fns(I)}$ & best-approximation operator \\
  % $\Fns \rightarrow \Fns(I)$ \\
  $\COp{I}$ & cluster operator \\
  $\Clust(I)$ & cluster subspace $\ran \COp{I}$  \\
  $\ran {\mathsf T}$ & range of operator ${\mathsf T}$  \\
  $\ker {\mathsf T}$ & subspace mapped to zero by ${\mathsf T}$ \\
%  $\Phi_I^\alpha$ & conventional cluster functions  \\
  calligraphic ($\mathcal S$) &  subspace of $\Fns$ \\
  sans-serif ($\prS$) & linear operator (usually projector) \\
  $\Arr{A}{\phi}{B}$ & $A$ is domain, $B$ codomain of function $\phi$ \\
  $x \mapsto f(x)$ & function defined by action on argument \\
  $J \subset I$ & strict subset: $J\subseteq I$ and $J\neq I$ \\
  $\bigoplus$ & direct sum (possibly orthogonal) \\
  $\|\phantom{f}\|$ & Hilbert space norm, $\|f\|^2 = \inpr{f}{f}$ \\
  $\Span A$ & linear span of set of vectors $A$ \\
  $\square$ & end of proof
\end{tabular}
\caption{Special notations and notational conventions.}
  \label{tab:notation-1}
\end{table}
Before starting the development, we establish some terminology, notations,
conventions, many of which are recorded in Table~\ref{tab:notation-1}
for convenient reference. 

$\Site$ is a set of {\it sites}.
Ordinarily, these are atomic sites in a crystal structure,
but until Section~\ref{sec:symmetry}, no spatial aspect is involved.
A {\it configuration} $x$ assigns to each site $i$ a {\it species} $x_i$
from the species set $\Sps{i}$, so distinct sites may have distinct
sets of allowed species.
These are typically simply chemical species, but we only need
$\Sps{i}$ to be a finite set; it may include information such as charge state,
for example. $\Cfg$ is the set of all configurations.
Properties of the system are real or complex functions, e.g. $f$, of configuration.
We will often call these \textit{observables} for brevity,
and to avoid overworking the word \textit{function}.
The set $\Fns$ of all observables is a vector space.
Because we will be very interested in which particular configuration
variables an observable depends on, some elaboration on these simple ideas is needed.
For a set $I$ of sites, a {\it partial configuration}
$x\in \Cfg(I)$ assigns species only to sites in $I$; we also say $x$ is
{\it supported} on $I$. For a complete configuration, $x\in\Cfg$, its restriction
to $I$ is denoted $x_I$, generalizing the notation $x_i$.
An $x$ or $y$ without adornment is a complete configuration
unless context indicates otherwise.
Thus, we have maps $\Cfg \rightarrow \Cfg(I)$ given by restriction.
Dually, there are natural linear maps $\Fns(I) \hookrightarrow \Fns$ going the
other way, from the space of functions of partial configurations supported on $I$
(that is, $\Fns(I)$) into $\Fns$. Namely, from $f\in\Fns(I)$, we obtain
the function $x\mapsto f(x_I)$ in $\Fns$. Since this is injective (one-to-one),
we might as well consider $\Fns(I)$ to be a subspace of $\Fns$.
With that identification, $f$ in $\Fns(I)$ is also in $\Fns(J)$ for any superset
$J$ of $I$. If $f\in\Fns(I)$, we say $f$ is {\it supported} on $I$, and {\it the} support
of $f$, $\supp f$, is the smallest such $I$.
That is, $\supp f$ is the set of sites whose species $f$ really depends on.

%%%%%%%%%%%%%%% 
\section{Rough ideas}\label{sec:from-scratch}

This Section begins construction of a new cluster expansion formalism
based on a simple notion of {\it best-approximation} which will be axiomatized
in the next section.

%%%%%%%%%%%%%%% 
\subsection{Effective and reduced entities}
\label{sec:effective-reduced}

{\it Reduced} or {\it effective} descriptions, formulations, quantities,
Hamiltonians, actions, theories, and so forth
are a venerable way to deal with complex systems in physical science.
Roughly, the idea is to eliminate dependence on some independent variables,
in such a way that the resulting simplified picture of the physical situation
is as accurate as possible within the resulting limited level of detail. 
For instance, classical fluid dynamics describes the motion of fluids on the basis
of only macroscopic degrees of freedom, atomic physics usually ignores the internal structure
of atomic nuclei. Renormalization group methods are among the most sophisticated
exemplars of this general philosophy.

We propose to base the cluster expansion method on this effective/reduced idea.
Begin with a configuration function $f \in \Fns$. Potentially, $f(\x)$ depends
on all components of $x$; we seek an effective/reduced description of $f$,
to be denoted $\EOp{I}f$,
which has lost its dependence on variables with indices outside $I$.
The intuitive content of ``effective/reduced description'' in this context
is perhaps better captured by the term {\it best approximation}, i.e.,
$\EOp{I} f$ is to be the best --- in some sense --- possible approximation to $f$
which is supported in $I$.
This is still extremely vague.
To make it precise, in the next Section we shall write down conditions,
that is, axioms,
that the operators $\Arr{\Fns}{\EOp{I}}{\Fns(I)}$ should satisfy,
not as a logical inevitability, but as the most plausible and reasonable
filling out of our fuzzy notion of best approximation. 

%%%%%%%%%%%%
\subsection{A heuristic shortcut}

A more heuristic analysis may be a helpful warm-up.
Return to the basic question: does ${\mathcal E}(x)$,
energy as a function of configuration, have a natural decomposition
into building blocks, absent a mechanistic theory?
Write $\COp{I}{\mathcal E}$ for the hoped-for building block
associated with $I$. This must be a function only of the subconfiguration
in $I$, $\COp{I}{\mathcal E}\in \Fns(I)$, and the total energy must be
the sum of all the building blocks, that is
\begin{equation}
  \nonumber
{\mathcal E}(x) = \sum_I (\COp{I} {\mathcal E})(x).
\end{equation}
Of course, this is nowhere near enough to determine the $\COp{I}{\mathcal E}$'s.
So, we make a stronger demand that 
\begin{equation}
\nonumber
\EOp{I} {\mathcal E}
\defeq \text{``energy in } I\text{''}
= \sum_{J\subseteq I} \COp{I} {\mathcal E},
\end{equation}
the energy \textit{in} $I$ is the sum of all the building blocks in $I$.
At first sight, this seems a completely bad move.
Asking for the energy \textit{in} $I$ is asking for a mechanistic
theory behind the scenes. Even if we had such, we could not expect it
to cleanly assign parts of the energy to specific sets of sites.
However, we can rephrase that in a more informational/observational
kind of way.
First, replace ``energy in $I$'' by
``what the subconfiguration in $I$ tells us about the energy''.
This still does not make sense, nor is it even a quantity of energy.
With another little shift, we reach the idea of interpreting
$\EOp{I} {\mathcal E}(x)$ as the \textit{conditional expectation}
of the total energy, given the configuration in $I$.
This does make sense. The good news is that, given all the $\EOp{I}{\mathcal E}$,
the preceding display uniquely determines the $\COp{I}{\mathcal E}$.
The bad news is that a probability distribution over configurations must
be brought in. On reflection, though, that seems inevitable.
$\COp{I}{\mathcal E}$ could hardly be expected to have nothing to do
with the environment of $I$. We can, however, reduce it to knowing what that
environment is like on average.

The key point in the above argument is to define the cluster components
$\COp{I}{\mathcal E}$ somewhat indirectly, in terms of the $\EOp{I}{\mathcal E}$,
because these are easier to get a solid grip on.
In the following formal development we avoid a direct invocation of probability
by starting from very explicit requirements on the $\EOp{I}$ operators.
That we wind up in the same place anyway underlines the naturality
of the probabilistic setting.

%%%%%%%%%%%%%%%%%%%%%%%
\section{Axioms for best-approximation projectors and
  consequences}\label{sec:axioms-and-consequences}

Section \ref{sec:axioms} gives the system of axioms \Ax{A}--\Ax{D}.
Section \ref{sec:cluster-components} introduces cluster components as a derived concept.
Section \ref{sec:defending-axioms} argues informally for
accepting them as the only plausible precisification of the idea of
effective/reduced functions.
Section \ref{sec:probability} brings out the implicit probability and
Hilbert space structures.

\subsection{Axioms}\label{sec:axioms}

We postulate a system of operators
$\Arr{\Fns}{\EOp{I}}{\Fns(I)}$, one for each set of sites $I$.
However, for reasons that become clear soon, we omit the subscript from
$\EOp{\varnothing}$, writing it as simply $\EOp{}$.
These operators are subject to the following axioms.
In Section \ref{sec:independence} we will take up a secondary axiom of independent sites.
\begin{enumerate}
\item[\Ax{A}.]
  $\Arr{\Fns}{\EOp{I}}{\Fns(I)}$ is a linear projection {\em onto} $\Fns(I)$.
\item[\Ax{B}.]
$f \ge 0 \text{ and } f\not\equiv 0 \; \Rightarrow\; \EOp{} f > 0$.
\item[\Ax{C}.]
$\EOp{} \EOp{I} = \EOp{}$.
\item[\Ax{D}.]
$g \in \Fns(I) \;\Rightarrow\; \EOp{I} (gf)  = g \, \EOp{I}f$.
\end{enumerate}

Since we have no notion of orthogonality yet,
saying that $\EOp{I}$ is a {\it projection} merely means that applying it twice
in succession is the same as applying it once:
$\EOp{I}\EOp{I} = \EOp{I}\circ\EOp{I} = \EOp{I}$.
However, we will soon see that they really are orthogonal projectors.
Generally, we omit the composition symbol $\circ$,
except in cases of ambiguity. 

%%%%%%%%%%%%%%%%%%%%%%%%%%%

\subsection{cluster components and decomposition}
\label{sec:cluster-components}

The idea of the following definition is that we want a single set of
building blocks to build all of the $\EOp{I} f$.
In the case of energy, these building blocks should be somehow describable as
two-body potentials, three-body potentials, and so forth.
\begin{defn}  \label{def:cluster-decomposition}
The {\em cluster operators}
\begin{equation}
  \label{eq:COp}
\Arr{\Fns}{\COp{I}}{\Fns(I)}
\end{equation}
are implicitly defined by
\begin{equation}
\label{eq:COp-def}
\EOp{I} = \sum_{J\subseteq I} \COp{J}.
\end{equation}
Applied to a function $f\in\Fns$, these produce the
{\em cluster components} $\COp{I} f \in \Fns(I)$.
Collectively, the cluster components comprise the
{\em cluster decomposition} of $f$.
\end{defn}

Anticipating that our axioms will imply that $\EOp{\Site}$ is the identity
on $\Fns$, notice that taking $I\equiv \Site$ in (\ref{eq:COp-def}) yields 
$f = \sum_{J} \COp{J} f$, shows that the appellation \textit{decomposition} is
appropriate.

Def. \ref{def:cluster-decomposition} is a legitimate definition because it
determines the cluster operators (equivalently, cluster components of arbitrary $f$)
recursively:
First, $\COp{\varnothing} f$ is determined outright, since $\varnothing$ has no
proper subsets, hence (\ref{eq:COp-def}) gives
$\COp{\varnothing} f = \EOp{\varnothing} f$.
Now, to recurse on the subset relation, assume that $\COp{J} f$ is determined for every
{\em proper} subset $J\subset I$.
Then, simple rearrangement of the putative definition gives 
\begin{equation}
\label{eq:COp-f}
\COp{I} f = \EOp{I} f - \sum_{J\subset I} \COp{J} f.
\end{equation}
Since everything on the right-hand side is known, $\COp{I} f$ is therefore determined
by the equation.
Inductive arguments of this sort will appear several more times in the development.

In fact, one can write down a formula for the solution of this recursion,
using the notation
\begin{equation}
  \varepsilon(I) = (-1)^{|I|}  =
  \begin{cases}
    \phantom{-}1 & \text{ if } I \text{ has even size } \\
    -1 & \text{ if } I \text{ has odd size }.
  \end{cases}
\end{equation}
\begin{prop}[M\"obius inversion]
\begin{equation}
    \label{eq:C-by-Mobius}
    \COp{J} = \varepsilon(J) \sum_{K\subseteq J} \varepsilon(K) \, \EOp{K}.
  \end{equation}
\end{prop}
This is an example of M\"obius inversion\cite{van-Lint+Wilson}.
The Proposition follows from (\ref{eq:Mobius-duality}) in Appendix \ref{app:Mobius},
where a self-contained exposition is given.

%%%%%%%%%%%%%%%%%%%%%%%%%%%

\subsection{Defending the axioms}\label{sec:defending-axioms}

%\subsubsection{Probability}
%\label{sec:probability}

Now we suggest some informal arguments for axioms \hbox{\Ax{A}--\Ax{D}}.
This is where we probe our intuition about what to see what are the
crucial features of the idea of {\it best approximation}.
The axioms are, however, the official rules and there is no more sneaking
intuition in by the back door.
Of course, we may later decide that we need to try a different axiomatization.
That will not happen in this paper, although the following section will
supplement \Ax{A}--\Ax{D} with a secondary independence axiom.

Begin with \Ax{A}, which is equivalent to: 
(i) the range of $\EOp{I}$ is in $\Fns(I)$,
(ii) $\EOp{I} g = g$ for every $g \in \Fns(I)$, and
(iii) $\EOp{I}$ is a linear operator. We defer (iii).
The argument for (i) and (ii) is as follows.
Suppose $f\in\Fns(I)$. Then, $\EOp{I}f$ means to
formally introduce additional variables that $f$ does not actually depend on
(one might be the position of Pluto), and then forget them again.
It would be absurd for that to result in anything except $f$ itself.

Now, (ii) above means that
for any constant $c$, \hbox{$\EOp{I} cf = c\EOp{I} f$},
and $\EOp{I}(f+g) = \EOp{I} f + \EOp{I} g$.
Since there are certainly effective/reduced situations
where this fails, such as nontrivial renormalization grou calculations,
this is not as obvious.
However, even the scalar multiple relation fails for these situations,
while it is pretty evident that we want that.
To be concrete, if the energy for every complete configuration is doubled, then
every effective/reduced energy should surely double as well.

This takes care of axiom \Ax{A}. Note a consequence for $\EOp{}$.
\Ax{A} implies that if $f$ is a constant $c$, then the effective/reduced $f$
depending on no variables is simply $f$ itself.
With linearity, that is the same as $\EOp{} 1 = 1$.

We argue for \Ax{B} by noting that it is equivalent to a continuity condition,
namely, 
\begin{equation}
\label{eq:continuity}
|g(\x)-f(\x)| < \epsilon \text{ for every } x \;\Rightarrow\;
|\EOp{} g - \EOp{} f| < \epsilon.
\end{equation}
It is easy to see that this implies \Ax{B}. Simply take $g$ constantly equal to the
average of the minimum and maximum values attained by $f$.
In the other direction, suppose $g < f + \epsilon$.
By \Ax{B} and \Ax{A}, $\EOp{} g < \EOp{} f + \epsilon$.
Now use $|g-f| = \max(g-f,f-g)$.

Suppose we arrive at an effective/reduced function by eliminating
some variables, and then we eliminate some more.
We might expect that the result is the same as if we had elimated all of
them at once. This would be $J\subseteq I \;\Rightarrow\; \EOp{J}\EOp{I} = \EOp{J}$.
\Ax{C} is a weak version of this, for just $J=\varnothing$.

Finally, we turn to \Ax{D}, which is the most awkward to argue for.
If $f$ is an interesting function of configuration, say energy, then it is not
so difficult to form some idea of what \Ax{A}, \Ax{B}, and \Ax{C} are saying.
\Ax{D}, however, contains a second function, $g$.
To motivate it, we start with a special case.
For $A\in\mathsf{Sps}$, define $\chi_A^i(\x)$ to be one if ${\x}_i\in A$, zero otherwise.
Now, consider $\chi_A^i f$.
The factor $\chi_A^i$ acts as a {\em switch}, turning $f$ {\em on}
only when $\sigma_i$ is in $A$.
When the switch is on, it is just as though the factor $\chi_A^i$ were absent.
For $I$ containing $i$, we take a configuration $\sigma$ with $\sigma_i\in A$,
and ask ourselves what $(\EOp{I} \tilde{f})(\sigma)$ should be.
It is supposed to be some sort of effective/reduced description of $\chi_A^i {f}$
using the limited data $\sigma_I$. But this data contains definitive proof
that the switch is on. In such case, we should agree that the presence of the
switch is irrelevant, that is,
\hbox{$\EOp{I} (\chi_A^i f) = \EOp{I} f = \chi_A^i \EOp{I} f$}
at configurations with \hbox{$\sigma_i\in A$}.
Accepting that, since $1-\chi_A^i$ is also a switch, we also have
\hbox{$\EOp{I} ((1-\chi_A^i) f) = (1-\chi_A^i) \EOp{I} f$}
when \hbox{$\sigma_i\not\in A$}.
By linearity (axiom \Ax{A}),
\hbox{$\EOp{I} (\chi_A^i f) = \chi_A^i \EOp{I} f$} without restriction.
Finally, since any function in $\Fns(I)$ can be expressed as a
sum of products of switches, \Ax{D} follows.

The system of axioms \Ax{A}--\Ax{D} might appear to have considerable freedom
in the $\EOp{I}$. That is illusory.
This subsection shows that $\EOp{}$ (that is, $\EOp{\varnothing}$)
alone determines a Hilbert space structure on $\Fns$, and the $\EOp{I}$
are then forced to be orthogonal projections onto the subspaces $\Fns(I)$.
This provides a more concrete structure to work with in the following
and is a step toward recovering the conventional formulation.

%%%%%%%%%%%%%%%%%
\subsection{Probabilistic interpretation}\label{sec:probability}

As we have already telegraphed, the best-approximation idea fixed by
\Ax{A}--\Ax{D} turns out to be probability in disguise.
This subsection carefully develops this rephrasing:
\begin{prop}
  \label{prop:prob-interp}
  Axioms \Ax{A}--\Ax{D} are equivalent to
  \newline
 1. $\Expct$ is a (full support) probabilistic expectation operator 
  \newline
 2. $\EOp{I}$ is the corresponding conditional expectation operator
  given the subalgebra $\Fns(I)$ of $\Fns$.
\end{prop}
We also take up the Hilbert space structure induced on $\Fns$ by $\EOp{}$
in Section~\ref{sec:P-Hilbert-space}, and use it heavily from that point on.
Probability notions are not necessary for that, but allow us to
talk about it in a second language.
The probability distributions we obtain here are {formal}.
Partly for this reason, we somewhat downplay explicitly probabilistic language.
In a concrete setting one may have already one or more relevant probability
distributions, e.g., Gibbs distributions with varying temperature and
chemical potentials. Introducing a cluster expansion in such a situation,
the corresponding probability distribution may not correspond to anything
``real''.
However, pragmatically, $\EOp{}$
is misguided if it does not provide a sensible \textit{reference distribution}.
A systematic development of general probability theory prioritizing expectation,
as we do, is given by Whittle\cite{Whittle}.

%%%%%%%%%%%%%%%%%%%%
\subsubsection{Basic probability}

We recall some probability concepts in preparation for rephrasing
the axiom system according to Prop.~\ref{prop:prob-interp}.

Over a discrete set such as $\Cfg$, a probability distribution is defined
simply by assigning a $P(x)\ge 0$ to each $x\in\Cfg$ in such a way that that
$\sum_{x\in\Cfg} P(x) = 1$.
If \hbox{$P(x) \neq 0$} for every $x$, we say $P$ has full support. This is
the case we are interested in.
We overload the notation so that
$P(A) = \sum_{x\in A} P(x)$ is the probability of a subset $A$ of $\Cfg$,
an \textit{event} in probability parlance.
We entend the notation once more to {\it expectations} of functions on the
outcome space $\Cfg$ (\textit{random variables}):
\begin{equation}
  \label{eq:expectation-def}
P(f) = \sum_{x\in\Cfg} P(x) f(x).  
\end{equation}
This is unorthodox notation.
A more usual notation for the expectation of $f$
would be something like $\EOp{}(f)$.
We will be able to revert to that notation as soon as the identification
promised by Prop.~\ref{prop:prob-interp} is established.
The abuse of notation committed in (\ref{eq:expectation-def}) is justified by the fact that
the probability of a set or single configuration equals the expectation of
its characteristic function. Defining
\begin{equation}
  \label{eq:is-it}
  \isit{x}(y) =
  \begin{cases}
    1 & y\equiv x \\
    0 & \text{otherwise}
  \end{cases}
\end{equation}
we can link the two uses of `$P$' in a different way:
\begin{equation}
P(x) = P(\isit{x}).
\end{equation}
This simple idea of reading an element of $\Cfg$ as an element of $\Fns$
will be surprisingly useful.

%%%%%%%%%%%%%%%%%%% 
\subsubsection{\texorpdfstring{$\EOp{}$ is an expectation operator}{}}

The first step in demonstrating Prop.~\ref{prop:prob-interp} is now simple.
$\EOp{}$ obeying \Ax{A} and \Ax{B} determines a probability $P$ on $\Cfg$ and vice versa.
Using the abuse of language already introduced, we write simply,
albeit somewhat cryptically,
\begin{equation}
\EOp{} = P.  
\end{equation}
That is, given $\EOp{}$, defining $P(x) = \EOp{}\isit{x}$ gives 
a full-support probability distribution:
$P(x) > 0$ by \Ax{B}, while \hbox{$\sum_x P(x) = \EOp{}(\sum \isit{x}) = \EOp{} 1 = 1$}
by \Ax{A}.
Going the other direction, if $\EOp{}f$ is defined via (\ref{eq:expectation-def}),
it satisfies \Ax{A} because $P(f)$ is a constant and $P(1)=1$, and
$\EOp{}$ satisfies \Ax{B} because \hbox{$P(x) > 0$}.
What \Ax{C} and \Ax{D} say in the case $I\equiv \varnothing$ is already subsumed by \Ax{A}.

%%%%%%%%%%%%%%%%%
\subsubsection{\texorpdfstring{$\EOp{}$ induces a Hilbert space structure}{}}
\label{sec:P-Hilbert-space}

The probability $P$ also naturally equips $\Fns$ with a Hilbert space structure,
through the definition
\begin{equation}
  \label{eq:inner-prod-def}
\inpr{f}{g} \defeq \sum_{x\in\Cfg} P(x) f(x)^* g(x) = P(f^* g)  
\end{equation}
of a positive-definite inner product. Here, $f(x)^*$ is the complex conjugate
of $f(x)$. We are allowing complex functions, although that is not necessary. 
Many probabilistic concepts can be recast in Hilbert space terms through the
use of this inner product. The expectation of $f$ is $P(f) = \inpr{1}{f}$,
where $1$ denotes the constant function with value $1$,
while its variance is
\hbox{$\Var f = P(|f|^2) - |P(f)|^2 = \|f\|^2 - |\ilinpr{1}{f}|^2$}.

Finally, consider a vector subspace ${\cal V}$ of $\Fns$ containing the
constants, and the orthogonal projection $\prV$ onto $\cal V$.
Then, $\prV f$ is a best approximation to $f$ by a random variable from $\cal V$
in the following probabilistic sense:
the residual $f-\prV f$ has mean zero, and
as $g$ ranges over $\cal V$, $\Var (f-g)$ is uniquely minimized at $g = \prV f$.
Suppose further that $\cal V$ is not only a vector subspace but also an algebra,
which means it is closed under multiplication ($g,h\in{\cal V} \Rightarrow gh\in{\cal V}$).
With a finite outcome space $\Cfg$, this implies that $\cal V$ contains every function
that can be made from its elements by any means (see discussion of $\sigma$ in
Sec.~\ref{sec:sketch-encoding}).
In such a case, probabilists call
$\prV f$ the {\it conditional expectation} of $f$ given $\cal V$.

%%%%%%%%%%%%%%%%%
\subsubsection{\texorpdfstring{The $\EOp{I}$ are orthogonal (self-adjoint) projections}{}}

Proof of Prop.~\ref{prop:prob-interp} is finished by showing that the orthogonal
projection onto $\Fns(I)$, and nothing else, can play the role of
$\EOp{I}$ in \Ax{A}--\Ax{D}.
\begin{lem}
  \label{lem:E-orthoproj}
  Axioms \Ax{A}--\Ax{D} imply that $\EOp{I}$ is an {\em orthogonal},
  that is, self-adjoint, projection onto $\Fns(I)$.  
\end{lem}
\begin{proof}
Since \Ax{A} already says $\EOp{I}$ is a linear projection, only
self-adjointness need be shown.
On the one hand,
\begin{align}
       \inpr{\EOp{I} g}{f}
& \stackrel{(\ref{eq:inner-prod-def})}{=}  \Expct (f {\EOp{I} g^*})
         \stackrel{{\textbf C}}{=}  \Expct \EOp{I} (f {\EOp{I} g^*})
         \nonumber \\ &
 \stackrel{{\textbf D}}{=}  \Expct ( \EOp{I}f \, \EOp{I} g^* ).
\label{eq:E-orthoproj-0}
\end{align}
On the other hand,
\begin{equation}
\nonumber
\inpr{g}{\EOp{I} f} = \inpr{\EOp{I} f}{g}^*
\stackrel{(\ref{eq:E-orthoproj-0})}{=} \Expct ( \EOp{I}f \, \EOp{I} g^* ).
\end{equation}
\end{proof}

Going the other way, it needs to be checked that orthogonal projections
actually satisfy all the axioms.
All but \Ax{D} are immediate. Thus, assume $g\in\Fns(I)$.
Then, since $\Fns(I)$ is closed under products (it is an algebra),
$\EOp{I}(gf)$ and $\EOp{I}f$ are also in $\Fns(I)$, and to show they are equal
only requires testing their inner products with an arbitrary $h\in\Fns(I)$:
\begin{align}
  \inpr{h}{\EOp{I}(gf)}
& {=} \inpr{\EOp{I} h}{gf}
\stackrel{h\in\Fns(I)}{=} \inpr{h}{gf}
\stackrel{\ref{eq:inner-prod-def}}{=} \inpr{g^*h}{f}
                          \nonumber \\
& \stackrel{g^*h\in\Fns(I)}{=} \inpr{g^*h}{\EOp{I} f}
 \stackrel{\ref{eq:inner-prod-def}}{=} \inpr{h}{g \EOp{I} f}.
    \nonumber
\end{align}
The first equality following by orthogonality of $\EOp{I}$.

%%%%%%%%%%%%%%%%%%%%%%% 
\subsubsection{Marginalization}\label{sec:marginalization}

The last main probability concept to bring in is {\it marginal distributions}.
Consider a system of projections $\EOp{I}$ satisfying \hbox{\Ax{A}--\Ax{D}}.
Fix a set $L$ of sites.
Restricted to $\Fns(L)$ --- understood as functions on $\Cfg(L)$ ---
the $\EOp{I}$ with $I\subseteq L$ still satisfy the axioms.
Therefore, just as in the previous few subsections, there is an induced Hilbert space
structure on $\Fns(L)$ and corresponding probability distribution
$\Arr{\Cfg(L)}{P_L}{(0,1]}$. This is a {\it marginal} distribution.
For $f,g\in\Fns(L)$,
\begin{equation}
\EOp{} (f^* g) = \inpr{f}{g} = \sum_{x\in\Cfg(L)} P_L(x) f(x)^* g(x).  
\end{equation}
Evaluate the marginal probability of
the partial configuration $x\in\Cfg(L)$ as
\begin{equation}
  P_L(x) = \EOp{} \isit{x} = \sum_{y\in\Cfg} P(y) \isit{x}(y)
  = \sum_{\substack{y\in\Cfg \\ y_I=x}} P(y).
\end{equation}
Like any other observable in $\Fns(L)$,
$\isit{x}$ simply ignores species at the other sites
when interpreted in $\Fns$:
\begin{equation}
\isit{x}(y) =
  \begin{cases}
    1 &
    \forall i\in \dom x, \, y_i = x_i
    \\
    0 & \text{otherwise}
  \end{cases}
\end{equation}
Essentially, this identifies the partial configuration $x$ with the
{\em set} of configurations which match it.
We could write
\begin{equation}
P_I(x) = P(x),  
\end{equation}
where on the right-hand side, `$x$' denotes that set and $P(x)$ is just
its ordinary probability under $P$.
Thus, the $I$ subscript above is dispensible.
However, in the case where $L$ is a single site, we will do so, and
with the special notation
\begin{equation}
  \label{eq:p_i}
p_i \equiv P_{\{i\}}.
\end{equation}
Suppose, for example, that Cu (copper) is in $\Sps{i}$. Then
\begin{equation}
  p_i(\text{Cu})
  = P_{\{i\}}(\text{Cu at } i)  
  = P(\text{Cu at } i).  
\end{equation}
The point here is that ``Cu at $i$'' is an acceptable partial configuration
(equivalently, set of configurations), so $P(\text{Cu at } i)$ makes sense.
But $P(\text{Cu})$ does not. Somewhere, it must be recorded that the magnesium
ion under discussion is at site $i$.
In this special case, we choose to retain it as a subscript on $p_i$
(lower-case, note).
Then, we can consider
$p_i$ as simply a probability distribution over the {\em set} $\Sps{i}$.

%%%%%%%%%%%%%%%%%%% 
\subsubsection{Internalizing evaluation, or, hats}

The operators $\EOp{I}$ fit very naturally into the Hilbert space structure of $\Fns$.
We can further rephrase even evaluation of an observable at a configuration as an
inner product. This will be very useful.
In Sec.~\ref{sec:modeling}, we represent models of $f$ based on
a sampling in terms of the sampled configurations, cast as observables.
Also, in the presence of symmetry, one wants to average a configuration
over a group action. The most reasonable way to make sense of that is to
realize the configuration as a special sort of observable.

Actually, the first step has already been taken in associating the
observable $\isit{x}$ with configuration $x$. The only problem is
that the normalization is not quite right for our purposes.
Therefore, we introduce an alternative normalization and notation:
\begin{equation}
  \label{eq:hat}
\hat{x} \defeq \Prob(x)^{-1} \isit{x}.  
\end{equation}
Here, $x$ may be a partial configuration.
For a complete configuration, we have immediately
\begin{equation}
  \inpr{\hat{x}}{f}
  = P(x)^{-1} \sum_{y\in\Cfg} P(y) \isit{x}(y) f(y)   
  = f(y),   
\end{equation}
since $\isit{x}$ reduces the sum to a single term.
The generalization to partial configurations follows by
the technique of relativization employed in Sec.~\ref{sec:marginalization}.
\begin{prop}\label{prop:hats}
For a partial configuration \hbox{$x_I\in\Cfg(I)$},
\begin{equation}
  \label{eq:EI-as-inner-prod}
(\EOp{I}f)(x_I) = \inpr{\wh{x_I}}{f}.  
\end{equation}
Equivalently, for $x\in\Cfg$,
\begin{equation}
\label{eq:hat-conditional}
\EOp{I}\hat{x} = {\wh{x_I}}.
\end{equation}
In particular, in the important special case of a complete configuration $x$,
\begin{equation}
  \nonumber
  % \label{eq:evaluation-as-inner-prod}
f(x) = \inpr{\hat{x}}{f}.  
\end{equation}
\end{prop}

%%%%%%%%%%%%%%%%% 
\section{Independent sites}\label{sec:independence}

With the aid of an additional condition, \IS, this section 
makes contact with the conventional cluster functions,
among other things. 
The conventional formalism corresponds to a probability distribution $P$
under which distinct sites are independent.
This notion is explained in Section~\ref{sec:Indep-axiom}.
Proposition~\ref{prop:independence} in Section~\ref{sec:Indep-disguises}
gives a large number of conditions equivalent to \IS.
Some of these conditions are interesting because they can be used
to argue in the spirit of Section~\ref{sec:defending-axioms} for
the naturality of \IS. This is done in Section~\ref{sec:Indep-defense}.
Finally, Section~\ref{sec:Phis} shows that the cluster functions of
the conventional formalism give ONBs when \IS\ holds.
Section~\ref{sec:is-it-projected} constructs the cluster components
of configurations which can provide a useful and convenient alternative
to the cluster functions.

\subsection{A secondary axiom}\label{sec:Indep-axiom}

We need to define the relevant concept before stating the axiom.
\begin{defn}[Independence]\label{def:independence}
Subalgebras ${\mathcal A}$ and ${\mathcal B}$ of $\Fns$
 are {\it independent} (more precisely: $\EOp{}$-independent)
if $\inpr{f}{g} = \EOp{}(fg) = \inpr{f}{1}\inpr{1}{g}$ for all
$f\in{\mathcal A}$, $g\in{\mathcal B}$.
\end{defn}
Here, then, is the site-independence axiom:
\begin{enumerate}
\item[\IS.]
  \label{axiom:IS}
  $\Fns(I)$ and $\Fns(J)$ are independent whenever $I$ and $J$ are disjoint.
\end{enumerate}

\subsection{Elementary products}

Elementary products are a class of observables which are especially
important and useful when \IS\ holds. 
\begin{defn}\label{def:factorizable}
The observable $f\in\Fns$ is an {\it elementary product}
if it can be written as a product
\begin{equation}
  \label{eq:elementary-function-def}
f = \prod_i f_i, \qquad f_i\in\Fns(i)
\end{equation}
of observables, each depending on only a single configuration component.
Equivalently, the set of elementary products is
\begin{equation}
  \label{eq:elementary-functions}
\prod_i \Fns(i).
\end{equation}
\end{defn}

It is extremely important that $\Fns$ is linearly spanned by the 
elementary products: $\Fns = \Span \prod \Fns(i)$.
This means that every $f\in\Fns$ is expressible as a linear
combination of elementary products, but such expression is {\em not}
always unique.
A basis, generally many, can always be extracted from a spanning set. 
The $\isit{}$ observables (\ref{eq:is-it}), for example, provide
an {\em orthogonal} basis of elementary products.
$\isit{x}$ is an elementary product since it can be written as
$\prod_{i} \isit{x_{\{i\}}}$, or, in a more explicit notation,
\begin{equation}
\isit{x}(y) = \prod_{i} \delta(x_i,y_i).  
\end{equation}
Since an arbitrary $f\in\Fns$ can be written
\begin{equation}
  \label{eq:f-expanded-as-isits}
f = \sum_{x\in\Cfg} f(x) \isit{x},
\end{equation}
the $\isit{}$ observables span $\Fns$. Finally, the definition
\hbox{$\inpr{\isit{x}}{\isit{y}} = \sum P(x')\isit{x}(x')\isit{y}(x')$}
shows that they are orthogonal, for any $P$, and therefore form an
orthogonal basis.

An elementary product contains exactly one factor from 
$\Fns(i)$ for each $i$, however some of these factors could simply be constants.
This motivates consideration of the \textit{sets} (not vector spaces) of observables
\begin{equation}
\nonumber
\prod_{i\in I} \Fns^0(i).
\end{equation}
$\Fns^0(i)$ is defined to be the orthogonal complement of constants in  $\Fns(i)$:
\begin{equation}
  \label{eq:F0-def}
  \Fns(i) = \Cmplxs \oplus \Fns^0(i).
\end{equation}
To see why these are important, consider an elementary product
$f = \prod f_i$, and write $f_i$ as $\EOp{} f_i + (f_i - \EOp{} f_i)$,
a sum of a constant and a term in $\Fns^0(i)$.
Expand the product as
\begin{align}
  \label{eq:elementary-expansion}
  f
  &  = \prod_i [\EOp{} f_i + (f_i - \EOp{} f_i)
    ]
    \nonumber \\
&  = \sum_I \left(\prod_{j\not\in I} \EOp{} f_j \right)
  \prod_{i\in I}(f_i - \EOp{} f_i).
\end{align}
Thus, any elementary product decomposes into a sum of scalar multiples
of observables in $\prod_{i\in I} \Fns^0(i)$ for varying $I$.
when \IS\ holds, these latter are the cluster components $\COp{I} f$,
as will be shown.

%%%%%%%%%%%%%%
\subsection{Many equivalents of independent sites}\label{sec:Indep-disguises}

Proposition~\ref{prop:independence} below lists a large number of
conditions equivalent to the proposed axiom \IS. Proof is deferred to
Section~\ref{sec:proof-of-prop}.
Some of the conditions are simple translations from one mode of expression to
another, as for example with (b) and (c).
Some are important for practical computations,
such as (b), (c), (h), and especially (i).
The last of these leads directly to useful formulas and connects
to the conventional formulation.
Others seem of more theoretical interest, such as (f) and (g).
The arguments for axiomatic status of \IS\ are based on (j) and (g).
\begin{prop}\label{prop:independence}
  The following are equivalent.
  \newline
In these, $I$ and $J$ implicitly run over all subsets of $\Site$,
  and $x$ over all configurations. $\prod f_i$ is a generic elementary product.
  
  \begin{enumerate}[label=\textnormal{(\alph*) }]
\item \IS
\item
  With one-site marginal distribution $p_i$ as in (\ref{eq:p_i}),
  \begin{equation}
  P(x) = \prod_{i} p_i(x_i).
\end{equation}
\item 
$\inpr{\prod f_i}{\prod g_j} = \prod_i \inpr{f_i}{g_i}$
\item 
$\EOp{I} \left( {\prod f_j} \right)
= \left(\prod_{j\not\in I} \EOp{}f_j\right) \prod_{i\in I} f_i$
\item
  $I\cap J = \varnothing \;\Rightarrow\;
  \EOp{I}\EOp{J} = \EOp{}$
\item $\EOp{I}\EOp{J} = \EOp{I\cap J}$ 
\item $J \not\subseteq I \;\Rightarrow\; \EOp{I}\COp{J} = 0$
\item $ \COp{I}\COp{J}= 
  \begin{cases}
    \COp{I} & I=J \\
    0 & I\neq J
  \end{cases}
  $
  %%%%%
\item
  $\COp{I} \left(\prod f_k\right)
  = \left(\prod_{j\not\in I}\EOp{}f_j\right) \prod_{i\in I} (f_i-\EOp{}f_i)$
  %%%%%
\item
  $\Clust({I}) = \Span \prod_{i\in I}\Fns^0(i)$, where $\Fns^0(i)$
 (\ref{eq:F0-def}) is the orthogonal complement of the constants in $\Fns(i)$.
\item
  $\Fns = \bigoplus_I \Clust(I)$,
  this being an ordinary direct sum
  (that is, $\Clust(I)\perp \Clust(J)$ is {\em not} assumed).
  \end{enumerate}
\end{prop}

\subsection{Why take \IS\ as an axiom?}\label{sec:Indep-defense}

If the $\EOp{}$-operator we are dealing with happens to satisfy
\IS, computational simplifications ensue, as one would expect from
conditions (b) or (c) of the Proposition.
This is reason to have a preference for such $\EOp{}$-operators.
However, just as we did for axioms
\Ax{A}--\Ax{D} in Section~\ref{sec:defending-axioms},
we want to find the best arguments toward the conclusion that
our semi-choate notion of approximation structure really requires
that the condition \Ax{IS} be imposed.

The first argument is that the operators $\EOp{}$ should {\em always} lead to a
{\em simplification}, in the sense of reducing the support of functions.
That is, $\supp \EOp{I} f \subseteq \supp f$, whatever $f$ and $I$ may be.
This translates, almost immediately to $\EOp{I}\EOp{J} = \EOp{I\cap J}$,
condition (f) of Prop.~\ref{prop:independence}.
While it has some surface plausibility, this argument is not really
very compelling. After all, some effective/reduced descriptions of the
sorts invoked as motivation in Sec.~\ref{sec:effective-reduced} introduce
degrees of freedom which do not even exist in the original description.

Better arguments are to be found by turning attention to cluster components.
We did not try to axiomatize the cluster operators directly because it is too
difficult to see {\it a priori} what a full set of such axioms should be.
Let us return to the idea that the cluster operators $\COp{I}$ should extract 
{the building blocks} of an observable. Inherent in that idea is some degree of
{\it atomicity}, {\it stability}, and {\it uniqueness}.
Hence, it might well be reasonable to demand that if we start with
a building block $\COp{I} f$, for some particular $I$ and $f$, and
extract {\em its} building blocks, we simply get it back and nothing else.
That would mean $\COp{I}\COp{I} f = \COp{I} f$, while
$\COp{J}\COp{I} f = 0$ for $J\neq I$. In other words, Prop.~\ref{prop:independence}(h).

Another argument of similar flavor is based on uniqueness.
Def.~\ref{def:cluster-decomposition} of the cluster operators provide a unique
expression of any observable as a sum of building blocks in the cluster subspaces $\Clust(I)$.
However, that does not mean that there are not others. Indeed, that is exactly what
the direct sum condition (j) of the Proposition imposes. 
If it does not hold, then there are other expressions
(not picked out by the $\COp{I}$ operators) of an observable as such a sum.

%%%%%%%%%%%%%%%%%%%%%5
%%%%%%%%%%%%%%
\subsection{Proof of Proposition \ref{prop:independence}}
\label{sec:proof-of-prop}

The equivalence of (a) -- (e) follows
by recursion from the following Lemma, elementary proof of which
is omitted.
\begin{lem}\label{lem:independence}
  For given $I$, $J$, the following conditions are equivalent.
  \begin{enumerate}[label=\textnormal{(\arabic*) }]
  \item $I$ and $J$ are independent
  \item $\EOp{I}\EOp{J} = \EOp{}$
  \item
    $\EOp{I} (f_I f_J) = (\EOp{} f_J) f_I$
    whenever $\supp f_I \subseteq I$ and $\supp f_J \subseteq J$.
  \item $P(x_{I\cup J}) = P(x_I)P(x_J)$
  \item
$ \inpr{f_If_J}{g_Ig_J} = \inpr{f_I}{g_I}\inpr{f_J}{g_J} $
whenever \newline
\hbox{ $\supp f_I \subseteq I$} and so forth.
\end{enumerate}
\end{lem}

\medskip
\noindent (e) $\Leftrightarrow$ (f): Implication right-to-left is immediate.
For the reverse, it suffices to show
\hbox{$\supp f \subseteq J \implies \EOp{I} f = \EOp{I\cap J}$}.
Since products span, we may even suppose $f=f_{J\cap I}f_{J\setminus I}$, where
the factors are supported in $J\cap I$ and $J\setminus I$, respectively.
Thus, compute
\begin{equation}
  \nonumber
  \EOp{I} (f_{J\cap I}f_{J\setminus I})
\stackrel{ \Ax{D} }{=} f_{J\cap I}(\EOp{I} f_{J\setminus I})
\stackrel{(d)}{=} f_{J\cap I}(\EOp{} f_{J\setminus I}),
\end{equation}
and note that starting with $\EOp{I\cap J}$ instead of $\EOp{I}$ 
yields the same result by the same reasoning.

\medskip
\noindent (f) \implies (g):
\begin{equation}
  \nonumber
\EOp{I} \COp{J}
 = \EOp{I} \left(  \EOp{J} - \sum_{K\subset J} \COp{K} \right)
\stackrel{(e)}{=}  \EOp{J\cap I} - \sum_{K\subset J} \EOp{I} \COp{K}.
\end{equation}
Inductively, we assume the conclusion $\EOp{I}\COp{K}=0$ for $K\subset J$, so
\begin{align}
\sum_{K\subset J} \EOp{I} \COp{K}
&  \stackrel{i.h.}{=} \sum_{K\subset I\cap J} \EOp{I} \COp{K}
   = \EOp{I} \sum_{K\subset I\cap J} \COp{K}
 \nonumber \\
& = \EOp{I} \EOp{I\cap J}
 = \EOp{I\cap J}.
   \nonumber
\end{align}

\medskip
\noindent 

\medskip
\noindent (g) \implies (h):
By M\"{o}bius inversion formula (\ref{eq:C-by-Mobius})
\hbox{$J \not\subseteq I \implies \COp{I}\COp{J} = 0$} follows from (g).
We will appeal to this shortly.

Now, The cases $J\not\subseteq I$, $J \subset I$, and $J = I$
are mutually exclusive and exhaustive.
The first case is just the assumption. If $J\subset I$, then $I\not\subseteq J$, so
$\COp{J}\COp{I} = 0$. Taking the adjoint of each side of this equation yields
$\COp{I}\COp{J} = 0$. For the final case, consider
\begin{equation}
  \nonumber
  \COp{I}^2
  = \left( \EOp{I} - \sum_{J\subset I} \COp{J} \right) \COp{I}    
  =  \EOp{I}\COp{I} - \sum_{J\subset I} \COp{J} \COp{I}
\end{equation}
The first term in the final expression is $\COp{I}$, while in the sum,
each term $\COp{J}\COp{I}=0$ by the second case.

\medskip\noindent
(h)\implies (e):
Use (\ref{eq:COp-def}) to express $\EOp{I}$ and $\EOp{J}$ in terms of cluster operators.
Applying (g), only one term ($\EOp{}$) in the expanded product survives.

%%%%%%%%%%%%% 
\medskip\noindent
(d)\implies (i):
Use the definition \ref{def:cluster-decomposition} and induction.
$\EOp{I} \prod f_j = c \prod_{i\in I} f_i$ and
\begin{align}
\prod_{i\in I} f_i
& \stackrel{(\ref{eq:elementary-expansion})}{=} \prod_{i\in I}(f_i-\EOp{}f_i)+
\sum_{J\subset I} \left( \prod_{k\not\in J}\EOp{}f_k \right)
\prod_{j\in J}(f_j-\EOp{}f_j)
    \nonumber \\
&\stackrel{i.h.}{=} \prod_{i\in I}(f_i-\EOp{}f_i)+
\sum_{J\subset I} \COp{J} \prod_{i\in I} f_i.
    \nonumber 
\end{align}
The first uses the expansion in (\ref{eq:elementary-expansion}), 
and the second line follows by induction hypothesis. Now apply 
Def.~\ref{def:cluster-decomposition}.

\medskip\noindent
(i)\implies (j):
$\COp{I}$ maps an elementary product into
$\prod_{i\in I}\Fns^0(i)$, so
$\Clust(I)\subseteq \Span \prod_{i\in I}\Fns^0(i)$.
On the other hand, (h) also shows that every member of
$\prod_{i\in I}\Fns^0(i)$ is fixed by $\COp{I}$, hence is in $\Clust(I)$.

\medskip\noindent
(j)\implies (k):
Define
$P^\times(x) \defeq \prod_i p_i(x_i)$
We are not assuming that $P^\times$ is the same as $P$. They just have the
same one-site marginals. However, $P^\times$ provides a legitimate inner product
$\inpr{\phantom{f}}{\phantom{g}}_\times$ on $\Fns$, as in Section~\ref{sec:P-Hilbert-space},
and we may apply the already established equivalence
\hbox{(b) $\Leftrightarrow$ (d)},
showing that the $\Clust(I)$ are orthogonal under the new inner product.
Therefore, they are certainly linearly independent, a statement which does not
involve any inner product at all. Thus, $\Fns = \bigoplus \Clust(I)$.

\medskip
\noindent 
(k) \implies (h): %\newline
The tactic is to show that the projection operators onto the components
of the assumed direct sum,
temporarily denoted $\widetilde{\mathsf{C}}_I$,
satisfy the defining conditions of the cluster operators, and therefore
are equal to them.
This suffices because such a system of projections always obeys the equations
of (h).

Now, $\Fns(I) = \sum_{J\subseteq I} \Clust(J)$, but by assumption
the $\Clust(J)$ subspaces are linearly independent, hence
\hbox{$\Fns(I) = \bigoplus_{J\subseteq I} \Clust(J)$}.
Therefore $\EOp{I} f
= \widetilde{\mathsf{C}}_I f  
+ \sum_{J\subset I} \widetilde{\mathsf{C}}_{J} f$. By induction hypothesis,
$\widetilde{\mathsf{C}}_{J} f = {\COp{J}} f$ in the sum. 
Removing those tildes then shows that 
$\widetilde{\mathsf{C}}_{I} f = {\COp{I}} f$ by (\ref{eq:COp-f}).

%%%%%%%%%%%%%%%%%%%%%%
\subsection{Conventional cluster functions}\label{sec:Phis}

The so-called cluster functions of the conventional formalism
(or a minor generalization thereof) provide ONBs of the
cluster subspaces $\Clust(I)$ precisely when \IS\ holds.
With Prop.~\ref{prop:independence} in hand, the demonstration is
very straightforward.
For each site $i$, let $\{\varphi_i^\alpha\}$ be an ONB of $\Fns^0(i)$.
For each $I$, define
\begin{equation}
\Phi_I^{{\alpha}} \defeq \prod_{i\in I}\varphi_i^{\alpha_i}
\end{equation}
for every collection of indices \hbox{${\alpha} = (\alpha_i\,|\, i\in I)$}.
This construction is just like in Section~\ref{sec:critical-review} except that
we are not using the encoding function $\sigma_*$.
Then, the $\Phi_I^{{\alpha}}$ comprise an ONB of $\Clust(I)$.
This claim follows from Prop.~\ref{prop:independence} (c,k).
It is also clear from Prop.~\ref{prop:independence}, that we cannot have
such a factorized basis unless site independence holds.

%%%%%%%%%%%%%%%%%%%%
\subsection{Cluster components of a configuration}\label{sec:is-it-projected}

With \IS, tidy formulas are also available for the
cluster components of $\isit{}$ observables (\ref{eq:is-it}).
That such observables provide an orthogonal basis of $\Fns$
indicates the theoretical interest in such formulas, but are they
practical? In fact, as discussed in Section \ref{sec:modeling},
information about nontrivial observables generally comes from sampling.
In such circumstances, it is entirely natural and efficient to build
models/approximations of such observables by use of hat-functions.

Even without \IS, $\isit{x}$ factorizes as $\prod \isit{x_i}$.
With \IS, the probability $P(x)$ also factorizes, resulting in
\begin{equation}
\label{eq:hat-x-factorizes}
\hat{x} = \prod_i \wh{x_i}.  
\end{equation}
Here we use the abbreviation
\begin{equation}
  \label{eq:hat-x-i}
\wh{x_i} \equiv \wh{x_{\{i\}}},  
\end{equation}
$\{i\}$ being a set of sites with one member.

Now, according to Prop.~\ref{prop:independence} (d) and (i)
\begin{align}
  \EOp{I} \hat{x} &
                    = \wh{x_I}
                    = \prod_{i\in I} \wh{x_i},
\label{eq:E-x-hat}                      \\
\COp{I} \hat{x} &= \prod_{i\in I} (\wh{x_i}-1).    
\label{eq:CI-x-hat}                      
\end{align}
Formula (\ref{eq:CI-x-hat}) is very useful, because it provides practical means to
extract cluster components without need for conventional cluster functions.
For a more explicit formula, notice that
\begin{equation}
  \nonumber
\inpr{\wh{x_i}-1}{\wh{y_i}} =
  \begin{cases}
p_i(x_i)^{-1} - 1, & x_i=y_i \\    
 - 1, & \text{otherwise} 
  \end{cases}
\end{equation}
and therefore
\begin{equation}
\label{eq:CI-hat-ip}
(\COp{I}\hat{x})(y)
= (-1)^{|I|} \prod_{i\in I:y_i=x_i} \left( 1-\frac{1}{p_i(x_i)}\right).
\end{equation}

%%%%%%%%%%%%%%%%% 
\section{Taking stock}\label{sec:intermission}

This Section recapitulates the framework we have assembled and 
relates it to what follows.

\subsection{Assessing the framework}

We started with the problem of breaking arbitrary
observables into natural building blocks.
While the conventional cluster expansion offers a sort of building block,
the cluster functions $\Phi_I^\alpha$,
they depend on some arbitrary choices, and are not adapted to the observable of interest.
To improve on the situation, we sought general principles,
and wrote axioms for the best-approximation operators $\EOp{I}$.
Cluster operators $\COp{I}$ are obtained from those in a very natural and simple way
(Def.~\ref{def:cluster-decomposition}).
Our axioms turned out to be equivalent
to requesting a probability distribution over configurations, with the
$\EOp{I}$'s identified as conditional expectation operators.
Thus we succeeded in imposing a reasonable amount --- and a familiar sort ---
of structure. To obtain the sort of Hilbert space structure on $\Fns$
assumed by the conventional cluster expansion requires an
additional condition of independent sites. 
We suggested arguments that \IS\ should be considered an axiom, though the
case is not as clear as for \Ax{A}--\Ax{D}.

With that, we arrive at a Hilbert space structure which is the same as
presumed by conventional cluster expansion with a variable probability
distribution of species.
What, then, has been gained?
First, we have characterized cluster subspaces $\Clust(I)$ and cluster
components $\COp{I} f$ intrinsically.
This helps insulate even the conventional cluster formalism from the charge
of being meaningless curve-fitting.
At the same time, we have identified a generalization, obtained by dropping the
condition \IS. Should an alternative way of decomposing functions of configuration into
building blocks be proposed, the abstract formulation puts us in a better position
for determining whether it is essentially the same, and even what the prospects are for
the existence of such an alternative.

\subsection{On the rest of the paper}

As we move toward more applied concerns in the remainder of the paper,
we have a flexible set of tools for working with $\Fns$ as a Hilbert space:
the critical orthogonal direct sum decomposition
\hbox{$\Fns=\bigoplus\Clust(I)$} [Prop.~\ref{prop:independence}(k)],
and a means (\ref{eq:CI-hat-ip}) to implement the $\COp{I}$ projectors.
This puts us in a position to formulate algorithms
in an intrinsic manner, and possibly find alternative implementions which
are more efficient than the use of cluster-function ONBs.
The formula of Prop.~\ref{prop:independence}(i) for cluster components of elementary
functions, and especially (\ref{eq:CI-x-hat}) for hat-functions, hint at such possibilities.

The major application of cluster expansion ideas in materials science is to
construction of models of an unknown function on the basis of
sampling a relatively small selection of configurations.
Practical success requires that the function in question be
reasonably approximable using just a few cluster components.
Symmetry plays an important role in making this possible, and
that is the topic of Sections \ref{sec:symmetry}, where we also expand
the framework to include tensor-valued observables.
Section \ref{sec:modeling} then discusses the modeling itself.
Hilbert space geometry is central to the perspective offered there,
in contrast to the conventional formalism,
where it all but disappears.

%%%%%%%%%%%%%%%%%
\section{Symmetry}\label{sec:symmetry}

%%%%%%%%%%%%
\begin{table}[h]
  \centering
\setlength{\tabcolsep}{12pt}
\renewcommand{\arraystretch}{1.5}
\begin{tabular}{ll}
  $\Gp$ & a symmetry group of $\EOp{}$ \\
  $T,T',\ldots$ & element of $\Gp$, site permutation \\
  ${\mathcal V}$ & space of ``tensors'' \\
  $\Av$ &  group averaging operator \\ 
  $\rho$ & representation of $\Gp$ on ${\mathcal V}$ \\
  $ \stackrel{\Gp}{\sim}$ & ``is $\Gp$-related to'' \\
  $\Gp I$ & $\Gp$-orbit of $I$, namely $\{J\,|\, J\stackrel{\Gp}{\sim} I\}$ \\
  $\Gp_I$ & stabilizer subgroup of $I$ \\
$\ClSh$ & cluster-shapes ($\Gp$-orbits of clusters) \\
$\Clust[\mathcal{I}]$ & $\bigoplus \{\Clust(J)\,|\, J\in \mathcal{I} \}$ \\
%  $\Clust[\Gp I]$ & $\bigoplus \{\Clust(J)\,|\, J\stackrel{\Gp}{\sim} I\}$ \\
  $\COp{\mathcal{I}}$ & projector onto $\Clust[\mathcal{I}]$ \\
\end{tabular}
\caption{More notation. The notations for $\Gp$-equivalence, orbit, and
  stabilizer subgroup may be used with other classes of entities than sets of
  sites. 
  \label{tab:notation-2}}
\end{table}
%%%%%%%%%%%%

To this point, $\Site$ has been merely a {\em set} of sites
with no structure. In at least the usual applications, however, the sites are
ionic positions in some crystal structure. Symmetry considerations
therefore emerge, which are very important for the application of
cluster expansion methods. 

For many observables, such as energy, whenever $I$ and $J$ are related by a
space group operation of the parent crystal structure,
cluster component $\COp{I} f$ is the same function of partial configuration
on $I$ as $\COp{J} f$ is of partial configuration on $J$.
This greatly reduces the computational burden of finding approximations of
such cluster components.
For tensor observables, elastic, optical, piezoelectric properties, and so on,
the situation is more complicated since the \textit{values} of such observables
also transform nontrivially under the space group (rotations, in particular).
So far, we have dealt only with \textit{scalar-valued} observables,
namely those in $\Fns$. Beginning in this section, we widen the scope to include
\textit{tensor-valued} observables, treating the former as a special case
of the latter.

Thsection extends the framework from $\Fns$ to $\vals\otimes\Fns$, where
$\vals$ is a space of tensor values, and shows how symmetry can be
effectively exploited.
We have been using the term \textit{observable} in a very broad sense
which would include, for example, the $xx$ component of the dielectric
constant, or the density of chains of three adjacent Co ions along the
$x$ direction. Assuming the $x$-axis is not invariant under our symmetry
group $\Gp$, these kinds of observables are not of interest to us as
they involve an explicit breaking of the symmetry.
We are here really interested only in $\Gp$-\textit{equivariant}
observables. For scalar-valued observables, this is genuine invariance,
whereas for a tensor-valued observable, it means we get the same thing
if we rotate the tensor value of the observable or rotate the structure
and then evaluate the observable.
Proposition~\ref{prop:unfolding} shows how the subspace $\Av \vals\otimes\Fns$
of $\Gp$-equivariant $\vals$-valued observables can be ``compressed''.
For convenient reference, most of the new notations we introduce here
are listed in Table~\ref{tab:notation-2}.

%%%%%%%%%%%% 
\subsection{Permutation actions}\label{sec:permutations}

Let $T$ be a site permutation. We consider briefly the natural action
of $T$ on various types of entity. Permutations defined by space
group operations are most important. However, that restriction is
irrelevant to the preliminary considerations here.
By its nature, $T$ has an action on sites $i\mapsto Ti$ with inverse $T^{-1}$.
This extends simply to clusters as
$TI = \setof{Tj}{j\in I}$ and to configurations as $(Tx)_i = x_{T^{-1} i}$.
As here, we generally use a uniform and simple notation for the natural
action on any type of object, i.e., the specific action is determined by
context. The pattern for configurations,
which are functions on $\Site$, continues to higher types.
The action of $T$ on $f\in\Fns = (\Cfg \rightarrow \Cmplxs)$
is given by $(Tf)(x) = f(T^{-1} x)$.
For a map $\Fns\rightarrow\Fns$, such as $\EOp{I}$ or $\COp{I}$,
matters are a little different because $T$ has a natural action on
the target space as well as the domain:
$(T\COp{I})(f) = T(\COp{I}(T^{-1} f))$, or 
$(T\COp{I}) = T \circ \COp{I} \circ T^{-1}$.
If we do not indicate explicitly on which space each $T$ is acting,
composition symbols are sometimes needed to disambiguate.

%%%%%%%%%%%% 
\subsection{\texorpdfstring{Symmetries of $\EOp{}$}{}}\label{sec:E-symmetries}

Permutations are not much use to us unless they are symmetries of $\EOp{}$.
$T$ is a symmetry of $\EOp{}$ if \hbox{$T \circ\EOp{} = \EOp{}\circ T$}.
In fact, $T\circ\EOp{}$ is just $\EOp{}$ if we are dealing with scalar observables,
but the more general form is relevant to tensor observables which will be taken
up below. In the following basic result about symmetries of $\EOp{}$,
the equivalence of (a) and (b) is of interest, but normally we only care
about (c) and (d) as consequences of symmetry.
%%%%%%%%
\begin{prop}
  \label{prop:E-symmetries}
For a site permutation $T$, the following conditions are equivalent:
  \newline
  \textnormal{(a) }
  $T$ is a symmetry of $\Expct$
  \newline
  \textnormal{(b) }
  $T$ is a unitary operator on $\Fns$: \newline
  \hspace*{10mm} for all $f, g \in \Fns$, $\inpr{Tf}{Tg} = \inpr{f}{g}$
  \newline
  \textnormal{(c) }
  for every $I$,  $\EOp{TI} = T\circ\EOp{I} \circ T^{-1}$
  \newline
  \textnormal{(d) }
  for every $I$,  $\COp{TI} = T\circ\COp{I} \circ T^{-1}$
\end{prop}
\begin{proof}
 \noindent (c) $\Rightarrow$ (a): $\EOp{} = \EOp{\varnothing}$.

  \noindent (a) $\Rightarrow$ (b):
   $\ilinpr{Tf}{Tg}
   = \EOp{}[(Tf)(Tg)]  
   = \EOp{}[T(fg)] $
   \newline \hspace*{52pt}
 $ \stackrel{\text{(a)}}{=} \EOp{}(fg) {=} \ilinpr{f}{g}$.

\noindent (b) $\Rightarrow$ (c):
  $T(\EOp{I}f)$ and $\EOp{TI}(Tf)$ are both in $\Fns(TI)$. Show they are equal by
  showing that they have the same inner product with an arbitrary test function
  $g\in\Fns(TI)$:
\begin{align}
\inpr{g}{T\EOp{I}f}    
& \stackrel{\text{(b)}}{=} \inpr{T^{-1}g}{\EOp{I}f}    
                          {=} \inpr{T^{-1}g}{f}
                          \nonumber \\  &
    \stackrel{\text{(b)}}{=} \inpr{g}{Tf}    
{=} \inpr{g}{\EOp{TI} (Tf)}.    
\end{align}
  \noindent (d) $\Leftrightarrow$ (c):
  Definition
%  (\ref{eq:cluster-ops})
  of $\COp{I}$, M\"{o}bius inversion (\ref{eq:C-by-Mobius}).
\end{proof}

If site independence holds, then $T$ is a symmetry of $\EOp{}$ only if it permutes
sites with identical species distributions. Thus, some space group operations of the
parent crystal structure may fail to be symmetries of $\EOp{}$ if the disordered
composition has a spatial modulation. On the other hand, the vast majority of
site permutations are not induced by euclidean operations. 
However, that does not imply that they are totally devoid of practical
interest. They are, for example, relevant to the constructuion of
special quasirandom structures (SQS's\cite{Zunger+90}), as well as to
regular solution models.

%%%%%%%%%%%%
\subsection{Tensor-valued observables}

To this point, we have considered only $\Reals$- or $\Cmplxs$-valued
observables, with energy as the primary example of an important one.
Many important observables, however, are not of this sort.
Polarization and magnetization are vectors, a dielectric tensor is
a rank-2 tensor, and, elasticity needs a rank-4 tensor for complete description.
The relevant extension of our formalism is straightforward.

Generally, we consider observables taking values in a finite-dimensional
Hilbert space $\vals$. Such an observable, belonging to the
tensor product $\vals\otimes\Fns$, can be written as
\begin{equation}
  \label{eq:tensor-observable}
  \vec{f} = \sum_i \evec_i \otimes f^i,
\end{equation}
where the $\evec_i$ are an ONB of $\vals$ and \hbox{$f^i\in\Fns$}.
We put an arrow above $\vec{f}$ to indicate that it takes values in some
nontrivial space, and use the term \textit{tensor observable} rather loosely
for all such observables.
The inner product on $\vals\otimes\Fns$ has two layers:
\begin{equation}
\ilinpr{\vec{g}}{\vec{f}}_{\vals\otimes\Fns}
 = \sum_i \inpr{g^i}{f^i}_\Fns  
 = \sum_i \sum_{x\in\Cfg} P(x) {g^i(x)^*}{f^i(x)}.
 \nonumber
\end{equation}
We have put subscripts on the inner products to be clear. Generally, we
will leave these off, relying on the rule that if both entries are in
the same space ($\Fns$ or $\vals\otimes\Fns$), then the inner product is
the corresponding one, and if one entry is scalar and one tensor, then
the inner product is that for $\Fns$, as in
\hbox{$\ilinpr{g}{\vec{f}} = \sum_i \evec_i \ilinpr{g}{f^i}$}.

Ordinary euclidean tensors fit naturally into this scheme.
If $\vals$ is rank-2 tensors, for example, then the inner product of
two tensors $A$ and $B$ is $\sum A_{ij}^* B_{ij}$, i.e., the matrix entries
are components relative to an orthonormal basis. Rotations and reflections
respect this orthonormality.

Elements of the space group of a parent crystal structure are,
or (perhaps a better word) induce site permutations.
If $T$ is such an element, then it acts on a tensor observable as
\hbox{$T(\sum \evec_i f^i) = \sum (T\evec_i) Tf^i$}, where \hbox{$f^i\mapsto Tf^i$}
is just the permutation action from Section~\ref{sec:permutations}.
For us, $T$ is not considered a symmetry operation unless it is
also a symmetry of $\EOp{}$. This identifies a subgroup of the space
group that will be denoted $\Gp$ henceforth.

We do not give the action of $T$ on elements of $\Fns$ a special name,
but it is sometimes helpful to write the action of $\Gp$ on $\vals$ as
\begin{equation}
  \label{eq:rep-rho}
  T \evec_i = \sum_j \rho_{ji}(T) \evec_j,
\end{equation}
so the $\rho(T)$ are the representation matrices. This representation is
not assumed to be irreducible.

Operators previously introduced on $\Fns$ are now lifted to $\vals\otimes\Fns$
in a natural way: the $\vals$ part simply passes through unchanged.
For example, if $v\in\vals$ and $f\in\Fns$,
$\COp{I} (v\otimes f) = \COp{I} (v f) = v (\COp{I} f) = v\otimes (\COp{I} f)$.
As just demonstrated, we can also often omit writing the tensor product symbol
if its clear what kinds of things the factors (here $v$ and $f$) are.

%%%%%%%%%%%%%%%%%%%
\subsection{Averaging over the group}

$\Gp$ will be a symmetry group for the observables of physical interest.
That is, such an observable $\vec{f}$ commutes with $T\in\Gp$,
$\vec{f}(Tx) = T(\vec{f}(x))$.
This is the same as $T\vec{f} = \vec{f}$, equivariance of $\vec{f}$.
Thus, \textit{averaging over the group} $\Gp$, achieved by the operator
\begin{equation}
  \label{eq:Av}
\Av^\Gp = \Av \defeq \frac{1}{|G|} \sum_{T\in\Gp} T
\end{equation}
is important. Generally, we omit the superscript when averaging is over $\Gp$.
It will be necessary when averaging over a subgroup such as a stabilizer
subgroup (e.g., $\Av^{\Gp_I}$).
On any vector space on which $\Gp$ is represented unitarily,
$\Av$ is an orthogonal projection onto the subspace of $\Gp$-equivariant vectors.
Thus, we can move it from one side to the other of an inner product,
e.g., $\ilinpr{\vec{g}}{\Av \vec{f}} = \ilinpr{\Av\vec{g}}{\vec{f}}$.
However, such a switch is illegitimate for $\ilinpr{g}{\Av\vec{f}}$,
with $\vec{f}\in\vals\otimes\Fns$, but $g\in\Fns$. The following Lemma
gives a useful formula for such a case.
\begin{lem}
  \label{lem:tensor-equivariants}
The $\Gp$-average of a separable tensor observable $\evec_i\otimes g = \evec_i g$ is  
  \begin{equation}
  \label{eq:tensor-equivariants-1}
\Av (\evec_i\otimes g) = \sum_j \evec_j \frac{1}{|\Gp|}\sum_{T\in\Gp} \rho_{ji}(T) Tg.
\end{equation}
From this follows
\begin{equation}
  \label{eq:tensor-equivariants-2}
\ilinpr{g}{\Av \vec{f}} = \sum_i \evec_i \frac{1}{|\Gp|}\sum_{T\in\Gp} \rho_{ji}(T) \ilinpr{Tg}{f^j}  
\end{equation}
\end{lem}
\begin{proof}
Compute
\begin{equation}
  \Av (\evec_i g)
    = \frac{1}{|\Gp|} \sum_{T\in\Gp} T\evec_i\otimes Tg.
% &    = \sum_j \evec_j \frac{1}{|\Gp|}\sum_{T\in\Gp} \rho_{ji}(T) Tg   
   \nonumber
\end{equation}
Inserting (\ref{eq:rep-rho}) and rearranging yields (\ref{eq:tensor-equivariants-1}).
For (\ref{eq:tensor-equivariants-2}), insert (\ref{eq:tensor-equivariants-1}) into
the final expression of
  \begin{equation}
  \ilinpr{g}{\Av \vec{f}}
  = \sum_i \evec_i \ilinpr{\evec_i g}{\Av \vec{f}}
  = \sum_i \evec_i \ilinpr{\Av \evec_i g}{\vec{f}}.
\nonumber    
  \end{equation}
\end{proof}

%%%%%%%%%%%% 
\subsection{Cluster-shape subspaces}\label{sec:cluster-shapes}

According to Prop.~\ref{prop:independence}(k),
$\Fns = \bigoplus_I \Clust(I)$ is the orthogonal direct sum of all
the cluster subspaces. Tensoring with $\vals$ on both sides yields
\begin{equation}
  \nonumber
\vec{\Fns} = \vals\otimes\Fns =
  \bigoplus_I \vals\otimes\Clust(I) = \bigoplus_I \vec\Clust(I).
\end{equation}
We have here introduced an unobtrusive over-arrow notation to indicate
the tensor products, and the point is that this change threads cleanly through
the cluster decomposition. Symmetry operations have a more profound effect.
According to Prop.~\ref{prop:E-symmetries}, 
whenever \textit{$T$ is a symmetry of} $\EOp{}$,
\hbox{$\vec\Clust(TI) = \COp{TI}\vec\Fns
= (T \circ \COp{I}\circ T^{-1}) \vec\Fns
= T (\COp{I}\vec\Fns) = T \vec\Clust(I)$}.
Symmetries permute the cluster subspaces for clusters in the
$\Gp$-orbit \hbox{ $\Gp I  = \{J\,|\, J\stackrel{\Gp}{\sim} I\}$}
of $I$. Here, $J\stackrel{\Gp}{\sim} I$, read ``$J$ is $\Gp$-related to $I$'',
means that $J=TI$ for some $T\in\Gp$.
We use the term \textit{cluster-shape} for a $\Gp$-orbit of clusters,
and denote them by calligraphic letters, e.g. $\mathcal{I}$.
So, $\mathcal{I} = \Gp J$ for any $J\in\mathcal{I}$.

The natural decomposition \hbox{$\vec\Fns = \bigoplus_{\mathcal{I}} \vec\Clust[\mathcal{I}]$}
is now into {cluster-shape} subspaces
\begin{equation}
\label{eq:cluster-type-space}
\vec\Clust[\mathcal{I}] \defeq \bigoplus_{J\in\mathcal{I}} \vec\Clust(J).
\end{equation}
Each cluster-shape subspace $\vec\Clust[\mathcal{I}]$ is collectively
\hbox{$\Gp$-invariant}, but individual vectors in it generally are not equivariant.
To obtain the subspace of $\Gp$-equivariant vectors, we simply apply the
$\Av$ projector: 
\begin{equation}
  \label{eq:decomp-into-Av-clustershape-subspaces}
\Av\vec\Fns = \bigoplus_{\mathcal{I}\in\ClSh }\Av\vec\Clust[\mathcal{I}].  
\end{equation}
Because $\Gp$ acts transitively on the orbit of $I$,
a nonzero $\Gp$-equivariant vector $\vec{f}$ in $\vec\Clust[\Gp I]$ necessarily has
nonzero component in each subspace $\Clust(J)$.
However, it can be recovered from just one such component $\COp{I} \vec{f}$, as shown
by the following Lemma. The computational significance of this is
that $\COp{I}\vec{f}$ is a function of fewer variables, namely just the part of
the configuration on $I$.
%%%%%%%%%%%%
\begin{prop}\label{prop:unfolding}
The maps in the following diagram  
\begin{equation}
  \label{eq:unfolding}
  \begin{tikzcd}[column sep=4 em]
    \Av^{\Gp_I}\Clust(I)
    \arrow[r,"\sqrt{|\Gp I|}\Av", shift left=0.7 ex]
&    \Av^{\Gp}\Clust[\Gp I]
\arrow[l,"\sqrt{|\Gp I|}\COp{I}",shift left= 0.7 ex]    
  \end{tikzcd}
\end{equation}
are mutually inverse unitary equivalences between the
subspace of \hbox{$\Gp_I$-equivariant} vectors in $\Clust(I)$ and
the subspace of \hbox{$\Gp$-equivariant} vectors in $\Clust[\Gp I]$.
\end{prop}
\begin{proof}
  It suffices to show that the linear maps in (\ref{eq:unfolding})
  have ranges in the indicated spaces, and are isometric.
  We drop the superscript from $\Av^\Gp$ now.
  
Ranges are in the indicated spaces:\newline
\hbox{$T \Clust(I) \subseteq \Clust[\Gp I]$}, $\forall T\in\Gp$,
so $\Av \Av^{\Gp_I} \Clust(I) \subseteq \Av \Clust[\Gp I]$.
\newline
In the other direction, $\Av^{\Gp_I}$ and $\COp{I}$ commute, so
\hbox{
  $(\COp{I} \circ \Av)\Clust[\Gp I]
  \subseteq (\Av^{\Gp_I}\circ\COp{I})\Clust[\Gp I]
  = \Av^{\Gp_I} \Clust(I)$}.

Lower arrow in (\ref{eq:unfolding}) is isometric: 
Take $\vec{f}\in\Av\Clust[\Gp I]$.
Since $J\in\Gp I\,\Rightarrow\,
\|\COp{J}\vec{f}\| = \|\COp{I}\vec{f}\|$,
\begin{equation}
  \label{eq:unfolding-lemma-1}
\| \vec{f} \|^2  
  = \sum_{J\in\Gp I} \| \COp{J} \vec{f}\|^2  
  = |\Gp I|\, \| \COp{I} \vec{f}\|^2.
\end{equation}

Upper arrow in (\ref{eq:unfolding}) is isometric: 
Take \hbox{$\vec{f}\in\Av^{\Gp_I}\Clust(I)$}. \newline
By (\ref{eq:unfolding-lemma-1}),
\hbox{$\| \Av\vec{f}\|^2 = |\Gp I| \|\COp{I}(\Av\vec{f})\|^2$}.
Combine this with
\hbox{$|\Gp|^2 \, \|\COp{I}(\Av\vec{f})\|^2
= \|\sum_{T\in\Gp_I} T \vec{f} \|^2 = |\Gp_I|^2 \, \|\vec{f}\|^2  $}.
\end{proof}

Lemma~\ref{prop:unfolding} allows us to effectively replace $\Av\Clust[\mathcal{I}]$,
by the much smaller space $\Av^{\Gp_I}\Clust(I)$ for an arbitrarily selected $I$
in the cluster-shape $\mathcal{I}$.
This is advantageous for computation.
With reference to the decomposition (\ref{eq:decomp-into-Av-clustershape-subspaces}),
we then have
\begin{equation}
\label{eq:Fns-into-representative-cluster-subspaces}
\Av\vec\Fns \cong \bigoplus_{I\in\text{section}} \Av^{\Gp_I} \vec\Clust(I).  
\end{equation}
Here, ``section'' denotes a selection of one cluster from each cluster-shape.

%%%%%%%%%%%%%%%%%%%%%%%%%%%%%%%%%%% 
\section{Two points of computational efficiency}\label{sec:efficiency}

The following Section will discuss how to construct a model for an observable
of interest using pieces belonging to $\Gp$-invariant cluster-shape subspaces
$\Av\vec\Clust[\Gp I]$. That will be done without using the traditional
cluster functions $\Phi_I^\alpha$.
Here, we take a brief detour to demonstrate the computational practicality of our approach.
Thus, we suppose that we have a \textit{cluster observable}, i.e., a
vector in $\Av\vec\Clust[\mathcal{I}]$ for some cluster-shape $\mathcal{I}$.
How can we represent and use it efficiently without the cluster functions?

%%%%%%%%%%% 
\subsection{Folding and unfolding}\label{sec:folding-and-unfolding}

A function in $\vec\Clust(I)$ has $\dim\vals$ components and depends on
$|I|$ site-occupation variables, while a cluster observable in
$\Av\vec\Clust[\Gp I]$ depends on $|\Gp I||I|$ variables.
Lemma~\ref{prop:unfolding} shows how the latter is reducible to the former in
principle. Let us see how that can be done in more detail.
We have $\vec{f}\in\Av\vec\Clust[\Gp I]$ and need to evaluate it on
configuration $y$.
It is useful to introduce maps
$\Arr{\Cfg}{\fold^I_{ij}}{\Fns(I)}$ by
\begin{equation}
  \label{eq:fold}
  \fold^I_{ij} y
  \defeq \frac{1}{|\Gp_I|}\sum_{T\in\Gp} \rho_{ij}(T)\, \EOp{I}(Ty).
\end{equation}
Here is the main result.
\begin{lem}\label{lem:folding}
For $\vec{f} \in \Av \Clust[\Gp I]$,  
\begin{align}
\vec{f}(y)
% & \stackrel{(\ref{eq:fold})}{=}
& = \sum_{ij} \evec_i \inpr{\fold^I_{ij} \hat{y}}{\COp{I} f^j}  
    \nonumber \\
  & = \sum_{ij} e_{i}\rho_{ji}(T)\, [\COp{I} f^j](T{y})
    \label{eq:using-fold}
\end{align}
\end{lem}
\begin{proof}
Apply (\ref{eq:tensor-equivariants-2}) to
$\vec{f}(y)
 \stackrel{(\ref{eq:unfolding})}{=}
    |\Gp I| \inpr{\hat{y}}{\Av (\COp{I} \vec{f})}$.
\end{proof}

The component of $\vec{f}$ in $\Av^{\Gp_I}\Clust(I)$
is just $\sum_i \evec_i \COp{I} f^i$, so the representations
in (\ref{eq:using-fold}) are obtained directly from this component.
The second sum only involves evaluating a function,
$[\COp{I}f^j](Ty)$, which is in $\Fns(I)$, on the
on the restriction of configuration $Ty$ to $I$.
The first form is interesting because 
$\ilinpr{\fold^I_{ij} \hat{y}}{\COp{I} f^j}$ involves a uniformity
of representation. Both the observable and configuration are
effectively in the small Hilbert space $\Fns(I)$ here.

%%%%%%%%%%%
\subsection{Storage and use of cluster observables}

Lemma~\ref{lem:folding} involves evaluating \hbox{$\COp{I} f^j \in \Fns(I)$}
on subconfigurations.
How may we represent such a function without using a cluster function basis?
Simply as a lookup table. Consider the efficiency of this method.
For the sake of simplicity, suppose that all species spaces $\Sps{i}$ have the
same size. Then, dimensions of relevant spaces are
\hbox{$\dim \Fns^0(i) = |\mathsf{Sps}|-1$}, 
\hbox{ $\dim\Clust(I) = (|\mathsf{Sps}|-1)^{|I|}$}, and
\hbox{$\dim \Fns(I) = |\mathsf{Sps}|^{|I|}$}.

Suppose $g$ is in $\Clust(I)$. Generally, we need to store at least
$\dim\Clust(I)$ numbers to represent it. 
Since the conventional cluster functions $\Phi^\alpha_I$ provide an ONB of
$\Clust(I)$, an efficient representation is obtained through the components
$\inpr{\Phi_I^{{\alpha}}}{g}$.
The hat-functions do not immediately provide a basis.
of $\Clust(I)$, but they do provide an orthogonal basis of the larger space $\Fns(I)$.
The function $g$ can then be represented by its components
$g(x) = \inpr{\hat{x}}{g}$ for $x\in\Cfg(I)$, of which there are $|\mathsf{Sps}|^{|I|}$.
Effectively, this is to consider $g$ as an element of $\Fns(I)$.
Compared to the conventional formalism, this is somewhat inefficient,
but not horribly so, unless $|I|$ is large.

There are countervailing considerations, however.
The whole point of having $g$ is to evaluate it on configurations.
With the second computational representation,
$g(x)$ is evaluated by simply stripping out $x_I$, the restriction of $x$ to
$I$, and looking up the corresponding value of $g$ in memory.
With the conventional computational representation, we have
\begin{equation}
  g(x) = \sum_{{\alpha}}
  \inpr{\hat{x}}{\Phi_I^{{\alpha}}}\inpr{\Phi_I^{{\alpha}}}{g}.
\end{equation}
It is necessary to first evaluate the factors $\inpr{\hat{x}}{\Phi_I^{{\alpha}}}$.
If we need to evaluate $g$ on many configurations, this becomes increasingly
costly.

%%%%%%%%%%%%%%%%%%%%%%%%%%%%%%%%%%% 
\section{Sampling, modeling, and Hilbert space geometry}
\label{sec:modeling}

\subsection{Defining a precise problem}

This section is concerned with how we can put the preceding formalism
to work in practical applications.
We assume some $\Gp$-equivariant scalar (e.g., energy) or tensor (e.g., dielectric tensor)
observable $\vec{f}$ of some configurationally disordered material, or class
of materials sharing a parent structure.
In terms both broad and vague, the problem is

\begin{itemize}
\item [P0:]
Find the most accurate and transferable model of $\vec{f}$, using limited sample data.
\end{itemize}

\noindent A mode can be \textit{transferable} with respect to many
system characteristics. System size is a particular one that we have
in mind here. The basis for thinking that the cluster expansion formalism
might lead to a model transferable in that sense is the idea of having components
in $\Clust(J)$ and $\Clust(I)$ which are simply related by a $\Gp$ operation
whenever $J$ and $I$ are so. It at least provides a plausible recipe for carrying
out the transfer. 
Of a proposed solution to P0, one could not hope to say
it is correct or incorrect. It may even be difficult to judge whether it is
better or worse than another. 
Therefore, while something like $\bm\alpha$ should remain a basic aspiration,
more precise, and probably limited, problems are needed to guide us.
We return now to a fixed-system-size framework, and put trasferability aside for
future analysis.
Two more precise versions are

\begin{itemize}
\item [P1:]
Using sampled data, find the best model of $\vec{f}$ in a
\textit{model subspace} $\ssp{M}$ of $\Gp$-equivariant observables.
%\noindent $\bm\alpha''$.
\item [P2:]
Using sampled data, find the best approximation of the orthogonal projection
$\prM \vec{f}$ of $\vec{f}$ into $\ssp{M}$.
\end{itemize}

In practice, the model space is
\begin{equation}
\label{eq:model-space}
{\ssp{M}} \defeq  \bigoplus_{\mathcal{I}\in\Small} \Av \vec\Clust[\mathcal{I}],
\end{equation}
where $\Small$ denotes some collection of cluster-shapes which are small,
in number of sites and/or spatial extent.
The orthogonal projector ${\prM}$ onto $\ssp{M}$ is
\begin{equation}
\label{eq:model-space-projector}
{\prM} = \sum_{\mathcal{I} \in \Small} \Av \circ \COp{\mathcal{I}}.
\end{equation}

P1 makes no reference to the Hilbert space structure of $\Fns$ at all.
P2 does, and in the process partially resolves the adjective \textit{best} in P1.
Namely, $\vec{f} = \prM\vec{f} + \prM^\perp\vec{f}$, where $\prM^\perp = 1-\prM$
is the complementary orthogonal projection to $\prM$. According to P2, our
target is $\prM\vec{f}$. 
This makes sense. $\prM\vec{f}$ provides the unique minimizer of
$\|\vec{f}-\vec{F}\|$ over all $\vec{F}\in\ssp{M}$, hence is the natural
Hilbert space notion of ``best approximation to $\vec{f}$ which lies in $\ssp{M}$''.

P2 still contains a \textit{best} that we shall have to interpret if we are
to assess any proposed solution.
One venerable option is to cash it out in statistical terms.
That entails a \textit{sampling distribution} according to which
our sample is \textit{randomly} drawn.
However, we are free to \textit{deliberately choose} the sample
so as to learn as much as possible about $\vec{f}$ with as little
effort as necessary.
For that reason, we do not take the statistical route.
The results of this section can help us see how to make those
choices more intelligently, so as to do more with a smaller sample.

Backing up a bit, we give our sample a name, $\Smpl$.
Sometimes this will be the collection of configurations paired with
the corresponding values of $\vec{f}$, and sometimes just the configurations
(which we will tend to write upper case to distinguish from generic configurations).
Now, it would be a waste to have two configurations in $\Smpl$ which are
$\Gp$-related, since $\vec{f}(TX) = T(\vec{f}(X))$, and we know how
$\Gp$ acts on $\vals$. Thus, we assume that $\Smpl$ contains more than one
representative of any $\Gp$ orbit.
Then, we know the values of $\vec{f}$ on every member of
\begin{equation}
\Gp\cdot\Smpl = \setof{TX}{X\in\Smpl,T\in\Gp}.  
\end{equation}
Or, in Hilbert space geometry terms, we know the orthogonal projection
of $\vec{f}$ into the subspace
\begin{equation}
\ssp{S} \defeq \Span \{\vals \otimes \wh{X} \,|\, X\in\Gp\cdot\Smpl \}
\end{equation}
of $\vals\otimes\Fns$.
In keeping with our convention, the orthogonal projector into $\ssp{S}$ is
denoted $\prS$. $\ssp{S}$ is evidently the maximal subspace on which we know
$\vec{f}$, and is called the \textit{shadow subspace}. The metaphor
is that $\vec{f}$ casts a shadow in $\ssp{S}$, and we are going to look for
a vector in the model subspace $\ssp{M}$ that casts a shadow as nearly
indistinguishable from that of $\vec{f}$ as possible.
Thus, our final version of the problem is

%\medskip
%\noindent $\bm\beta$.
\begin{itemize}
\item [P3:] Given the shadow subspace $\ssp{S}$ constructed from $\Smpl$
as described, seek $\vec{F}\in\ssp{M}$ such that
\begin{enumerate}[label=\textnormal{\arabic*$^\circ$}]
\item $\| \prS \vec{F} - \prS \vec{f} \|$ is minimum.
\item $\| \vec{F} \|$ is minimum, subject to 1$^\circ$.
\end{enumerate}
\end{itemize}
%\medskip

\noindent
The idea is that, of $\vec{f}$, we know only that it is $\Gp$-equivariant,
and its projection $\prS\vec{f}$. Therefore, we seek a model in $\ssp{M}$
matching it in these respects. However, it may not be possible to match
it perfectly, hence 1$^\circ$, and even then, there may be freedom
in the model, so 2$^\circ$ stipulates that idle structure be stripped out.
The given norm on $\Fns$ provides throughout our standard of comparisons.
We proceed to give the solution to this problem, prove its correctness,
and contrast it with the traditional approach to fitting a cluster expansion model.

%%%%%%%%%%%%%%%%5
\subsection{Solving the abstract problem}\label{sec:abstract-model}

The solution to problem P3 is given in Proposition \ref{prop:Up} below,
which shows that it is realized by a linear operator $\Up{\ssp{S}}{\ssp{M}}$
which does not depend on $\vec{f}$ at all, but only on $\ssp{S}$ and $\ssp{M}$.
In principle, the linear operator $\Up{\ssp{S}}{\ssp{M}}$ provides simultaneous
models for any number of observables measured on the sample \Smpl.
While it is not generally practical to construct
$\Up{\ssp{S}}{\ssp{M}}$ in its entirety, it can be evaluated
on an individual observable by (\ref{eq:min}).
The following abbreviations will be used:
\begin{alignat}{2}
\ssp{S}' & \defeq \prS {\ssp{M}}    
%=  \im \prS|_{\ssp{M}},
\qquad
& \prM_0 & \defeq \prM|_{\ssp{S}'}
% = \im \prM|_{\ssp{S}}
\nonumber \\                           
 {\ssp{M}}' & \defeq \prM \ssp{S} \qquad
& \prS_0 & \defeq \prS|_{\ssp{M}'}.
\end{alignat}
The preparatory Lemma \ref{lem:Up} clarifies the geometry,
showing that the orthogonal complement of $\ssp{M}'$ in the
model space $\ssp{M}$ is completely inert for modeling purposes, while
the orthogonal complement of $\ssp{S}'$ in $\ssp{S}$ cannot be
effectively modeled from $\ssp{M}$.
%%%%%%%%%%%%
\begin{lem}
  \label{lem:Up}
The following orthogonal direct sum decompositions hold:
\begin{subequations}
\begin{align}
  \ssp{S} & = \ssp{S}'  \oplus  \ker \prM|_{\ssp{S}}
            \label{eq:S-ortho-decomp}
  \\
  \ssp{M} & = \ssp{M}'  \oplus  \ker \prS|_{\ssp{M}}
            \label{eq:M-ortho-decomp}
\end{align}
\end{subequations}
Moreover, the restricted projections
\begin{equation}
   \nonumber
\Arr{\ssp{S}'}{\prM_\bullet}{\ssp{M}'} \text{ and }\;
\Arr{\ssp{M}'}{\prS_\bullet}{\ssp{S}'}
\end{equation}
are bijections, hence have well-defined inverses.
%$\prM_\bullet^{-1}$ and $\prS_\bullet^{-1}$.
\end{lem}
\begin{proof}
  For $s\in{\cal S}$,
  \begin{equation}
    \nonumber
    s\in \ker \prM
   \, \Leftrightarrow\, \forall v\in{\ssp{M}}, \, \inpr{s}{v}=0    
   \, \Leftrightarrow\, s \perp \prS{\ssp{M}}.
 \end{equation}
 This establishes (\ref{eq:S-ortho-decomp}).
It also shows that $\prM$ is injective on ${\cal S}'$, and
since any $v\in{\ssp{M}}'$ is the image of something in ${\cal S}'$,
bijectivity of $\prM_\bullet$ follows.
The rest follows by symmetry.
\end{proof}
%%%%%%%%%%%%
\begin{prop}\label{prop:Up}
There is a unique vector $\Up{\ssp{S}}{\ssp{M}} \vec{f}\in\ssp{M}$ satisfying conditions
\textnormal{1$^{\circ}$} and \textnormal{2$^{\circ}$}.
It is given by the linear operator
\begin{equation}
\label{eq:Up}
\Up{\ssp{S}}{\ssp{M}} = \Arr{\Fns}{\prS_\bullet^{-1} \prS'}{\ssp{M}'},
\end{equation}
where $\prS'$ denotes orthogonal projection onto \hbox{${\cal S}' = \prS \ssp{M}$}.
The model $\Up{\ssp{S}}{\ssp{M}} \vec{f}$ is alternatively characterized by
\begin{equation}
  \label{eq:min}
\Up{\ssp{S}}{\ssp{M}} f = 
\underset{h\in\prM\ssp{S}}{\arg\min} \|\prS f - \prS h\|.
\end{equation}
\end{prop}
%%%%%%%%%%%%
\begin{proof}
The preceding lemma says that \hbox{$\prS \ssp{M} = \prS \prM  \ssp{S}$}.
Hence, $\tilde{f}$
has the form $\prM g$ for some $g\in\ssp{S}$ because any component of $\tilde{f}$
orthogonal to $\prM\ssp{S}$ would increase its norm without changing its image under
$\prS$. Appeal to criterion 1$^\circ$ then establishes (\ref{eq:min}).

 Differentiation with respect to $g$ says that minimizing \hbox{$\|\prS f - \prS\prM g\|^2$}
 requires $\prM(\prS f - \prS\prM g)=0$.
 Equivalently, $\prS\prM g$ is the projection
$\prS' \prS f  = \prS' f$ of $f$ into $\prS\ssp{M}$.
By defintion of $\prS_\bullet$ (see Lemma), \hbox{$\prM g = \prS_\bullet^{-1} \prS'$}, which finishes
the proof, given (\ref{eq:min}).
\end{proof}

%%%%%%%%%%%%%%%%
\subsection{Application to the sampling context}
\label{sec:shadow-from-sampling}

Now we apply the abstract Proposition~\ref{prop:Up} to the context in which
we are interested, namely, where $\ssp{S}$ is determined by a sample $\Smpl$.
Concerning $\ssp{M}$, we only assume that its vectors are $\Gp$-equivariant.
There are essentially two aspects to this: the concrete form of the loss function,
given in Proposition~\ref{prop:model-from-sample}, and concrete parametrization of candidate
models in $\ssp{M}'$, discussed later.

\subsubsection{Loss function}
%%%%%%%
\begin{prop}\label{prop:model-from-sample}
  Assume
  \begin{itemize}
  \item $\ssp{S}$ is determined by $\Smpl$, no two members of which are $\Gp$-related
  \item $\Av \ssp{M} = \ssp{M}$, $\Av\vec{f} = \vec{f}$.
  \end{itemize}
  Then,
\begin{align}
& {\ssp{S}}  =
            \Span \setof{\vals\otimes T\hat{X}}{T\in\Gp,X\in\Smpl},
            \label{eq:prop-model-S} \\
& {\ssp{M}'} = \Span \setof{\prM (\vals\otimes \hat{X})}{X\in\Smpl}.
            \label{eq:prop-model-M} 
\end{align}
The model is given by
\begin{equation}
\nonumber
\Up{{\ssp{S}}}{{\ssp{M}}} \vec{f}
= \arg\min_{\vec{F}\in{\ssp{M}}'} L(\vec{F})
\end{equation}
with the concrete loss function
\begin{equation}
\label{eq:loss-concrete}
L(\vec{F}) =
  \sum_{Y\in\Smpl} P(\Gp Y) \left\| \vec{f}(Y) - \vec{F}(Y) \right\|_\vals^2.
\end{equation}
\end{prop}
\begin{proof}
Except for (\ref{eq:loss-concrete}),
this is a straightforward transcription of Prop.~\ref{prop:Up}.
As an ONB of $\ssp{S}$, we take
$\sqrt{P(Y)} \vec{e}_i \hat{Y}$ where $\vec{e}_i$ runs over an ONB of $\vals$,
and $Y$ over $\Gp\cdot\Smpl$.
$\|\prS(\vec{f}-\vec{F})\|_{\vals\otimes\Fns}$ is then
the sum of terms \hbox{${P(Y)} |\ilinpr{\vec{e}_i\hat{Y}}{\vec{f}-\vec{F}}|^2$}
over the indicated basis.
\end{proof}
%%%%%%%%%%

$L(\vec{F})$ differs from an ordinary least-squares loss function in that
each $\Gp$-orbit of configurations which is represented in $\Smpl$
makes a contribution weighted by its total $\EOp{}$-probability.
Also, it is completely determined by the conditions of problem P3.

\subsubsection{Model parametrization}

According to (\ref{eq:prop-model-M}), $\ssp{M}'$ may be
parameterized as
  \begin{equation}
  \label{eq:lambda-model}
  \vec{F}_\lambda \defeq \sum_{X\in\Smpl} \prM (\vec{\lambda}(X)\otimes\hat{X}),
\end{equation}
with a $\vals$-valued variational parameter $\vec{\lambda}(X)$ for each
\hbox{$X\in\Smpl$}.
Previous sections have discussed how to handle $\prM(\vec{\lambda}(X)\hat{X})$
computationally. See, in particular, (\ref{eq:CI-hat-ip}) and
Lemma~\ref{lem:tensor-equivariants}.
The parameterization in (\ref{eq:lambda-model}) might be redundant.
In particular, if $\Gp_X$ is nontrivial, then there are vectors
$0 \neq \vec{e}\in\vals$ for which $\Av (\vec{e}\hat{X}) = 0$.
However, that is harmless; it merely means that different $\lambda$
parameterize the unique $\vec{F}_\lambda$ minimizing $L(\vec{F}_\lambda)$.

On the other hand, if we parameterize candidate models as
\begin{equation}
  \vec{F}_e = \sum_{I\in\Small,\alpha} \vec{e}_I^\alpha \otimes \Phi_I^\alpha,
\end{equation}
using conventional cluster functions $\Phi_I^\alpha$ (or any other basis of $\ssp{M}$)
with coefficients $\vec{e}_I^\alpha\in\vals$,
then the minimizing function will be nonunique whenever $\ssp{M}'$ is not
all of $\ssp{M}$. A standard fix for such underdetermination is to
add a penalty term involving an adjustable scale parameter to the loss function.
Such ad hoc fixes are rendered unnecessary by parameterizing only $\ssp{M}'$,
as in (\ref{eq:lambda-model}).
Circumstances under which $\Smpl$ might be deliberately engineered so
that $\ssp{M}'$ is a proper subspace of $\ssp{M}$ are not hard to imagine.
For instance, one may be interested in a low density of vacancies, and
therefore use only configurations with no more than one vacancy in any
cluster-shape in $\Small$.

%%%%%%%%%%%%%%%%%%%%%%%%%%%%%%%%%%%
\section{Conclusion}

Our primary objective here was to advance understanding of
\textit{what} cluster expansion really means.
To that end, we proposed a system of axioms for best-approximation operators
and a definition connecting them to basic building blocks we call cluster
components. The resulting structure is compatible with the traditional
cluster expansion formalism in a special case, without the need for
the usual cluster functions.
Due to the importance of cluster expansion methods in practical computations,
such a foundational advance potentially enjoys a large multiplier effect.
The approach presented here to the problem of modeling from a sample
supports that possibility, and achieves two specific ends.
First, a precise version of the problem is fully integrated with the
fundamental formalism, in contrast to the traditional approach, where
it appears ad hoc. Second, it shows that even for practical computations
one can do without conventional cluster functions.
Indeed, the alternative parametrization using sample data has an advantage
in avoiding a potential underdetermination problem.

%%%%%%%%%%%%%%%%%%%%%%%%%%%%%%%%%%%
\begin{acknowledgments}
  This work was supported by the National Science Foundation under
  grant DMR-2011839. We thank Miguel Mercado and Ismaila Dabo for
  asking questions and pointing out obscurities.
\end{acknowledgments}

%%%%%%%%%%%%%%%%%%%%%%%%%%%%%%%%%%%

\appendix

\section{M\"obius duality}\label{app:Mobius}

Notation in this appendix is different from in the main text.
For some set $\Omega$ (not necessarily finite), we are concerned with functions
on the set ${\mathscr P}_{\text{fin}}(\Omega)$, with values in some vector space
$V$. Such subsets will be denoted $I$, $J$, $K$, etc., and use a function on them
written in the normal way, e.g., $f(I)$.
M\"obius inversion is a procedure applicable in a much wider range of
combinatoric settings\cite{van-Lint+Wilson} than discussed here, but this one
is more than adequate for our needs.

The equivalence (\ref{eq:Mobius-duality}) is our target; the path is organized
as two lemmas.
\begin{lem}
  \label{claim:mobius-1}
  \begin{equation}
    \label{eq:mobius-lem}
    \sum_{J\colon K\subseteq J\subseteq I} \varepsilon(J) =
    \begin{cases}
      \varepsilon(I) & K=I \\
      0 & \text{otherwise}
    \end{cases}
  \end{equation}
\end{lem}
\begin{proof}
  The cases $K\not\subseteq I$ and $K=I$ are obvious since the sum is
  empty in the first case and has only one term in the second.

Now consider $I' = I\cup \{i\}$, where $K\subseteq I$ and $i\not\in I$.
Subsets of $I'$ come in pairs of the form $(J, J\cup\{i\})$ with $J\subseteq I$.
Each such pair makes a contribution $1 + (-1) = 0$ to the sum in (\ref{eq:mobius-lem}).
\end{proof}

For convenience, define the {\it M\"obius transform}
\begin{equation}
  \label{eq:Mobius-transform}
({\Mop} f)(I) \defeq \sum_{J\subseteq I} f(J),  
\end{equation}
and a version of $\varepsilon$ rendered as a transformation
\begin{equation}
  \label{eq:phi-op}
({\hat{\varepsilon}} f)(I) \defeq \varepsilon(I) f(I).  
\end{equation}
Then, the composition $\hat\varepsilon \Mop$ is an involution:
%%%%%%%%%%%%%%%%
\begin{lem}
  \label{claim:Mobius-duality}
$(\hat{\varepsilon} \Mop)^2 = \Id$.
\end{lem}
\begin{proof}
  \begin{align}
    \Mop \hat{\varepsilon} \Mop f(I)
    &= \sum_{J\subseteq I} \hat{\varepsilon} \Mop f(J)
      \nonumber \\ &
                     =\sum_{J \subseteq I} \varepsilon(J) \sum_{K \subseteq J } f(K)
      \nonumber \\ &
                    = \sum_{K \subseteq I} f(K) 
                     \sum_{J: K \subseteq J \subseteq I} \varepsilon(J)
      \nonumber \\ &
                     \stackrel{\text{claim. \ref{claim:mobius-1}} }{=}
                     (\hat{\varepsilon} f)(I)
  \end{align}
\end{proof}
This lemma immediately implies the following {\it M\"obius duality} relation:
\begin{equation}
  \label{eq:Mobius-duality}
g = \Mop f \;\Leftrightarrow\; \hat{\varepsilon} f = \Mop (\hat{\varepsilon} g)    
  \end{equation}

  %%%%%%%%%%%%%%%%
%%%%%%%%%%%%%%%%%%%%%%%%%
%  \bibliography{HEO.bib}


%%%%%%%%%%%%%%%%%%%%%

\end{document}